\documentclass[reqno,a4paper,twoside,12pt]{amsart}

\usepackage{dsfont}
\usepackage[latin1]{inputenc}
\usepackage{graphicx}

\IfFileExists{mymtpro2.sty}%
	{\usepackage[mtpscr,subscriptcorrection]{mymtpro2}
		\let\tilde\wtilde
		\let\bar\wbar
  	\let\cup\cupprod
		\let\cap\capprod
		\def\bigcup{\bigcupprod\limits}
		
		\def\llshift{-2mu}}%
	{\def\llshift{-7mu}\usepackage{amssymb}}  

\newcommand{\wtilde}{\tilde}

  \newtheorem{theorem}{Theorem}[section]
  \newtheorem{corollary}[theorem]{Corollary}
  \newtheorem{proposition}[theorem]{Proposition}
  \newtheorem{lemma}[theorem]{Lemma}
\theoremstyle{definition}
  \newtheorem{definition}[theorem]{Definition}
\theoremstyle{remark}
  \newtheorem{remark}[theorem]{Remark}
  \newtheorem{myrems}[theorem]{Remarks}
  
  \newtheorem{myexs}[theorem]{Examples}

\newenvironment{remarks}{\begin{myrems}\begin{nummer}}%
    {\end{nummer}\end{myrems}}

   {\end{nummer}\end{myexs}}

\DeclareMathOperator{\supp}{supp}
\DeclareMathOperator{\spec}{spec}
\DeclareMathOperator{\dom}{dom}
\DeclareMathOperator{\tr}{tr}
\def\timesop{\mathop{\raisebox{-1.5pt}{\Large $\times$}}\limits}
\DeclareMathOperator{\one}{{\mathds{1}}}
\DeclareMathOperator{\e}{e}
\def\d{\mathrm{d}}
\let\i\undefined
\DeclareMathOperator{\i}{i\kern.5pt}

\newcommand{\cH}{\mathcal{H}}

\def\Chi{\raisebox{.2ex}{$\chi$}}

\newcommand{\NN}{\mathbb{N}}
\newcommand{\ZZ}{\mathbb{Z}}
\newcommand{\zd}{\ZZ^{d}}
\newcommand{\RR}{\mathbb{R}}
\newcommand{\CC}{\mathbb{C}}
\newcommand{\EE}{\mathbb{E}}
\newcommand{\PP}{\mathbb{P}}

\newcommand{\DD}{\mathds{D}}
\newcommand{\HH}{\mathds{H}}
\newcommand{\UU}{\mathds{U}}
\newcommand{\IZ}{\mathds{N}}

\newcommand{\sma}[2]{\tmatrix{#1 & #2}{#2 & -#1}}

\let\emptyset\varnothing
\renewcommand{\le}{\leqslant}
\renewcommand{\ge}{\geqslant}
\renewcommand{\leq}{\leqslant}
\renewcommand{\geq}{\geqslant}

\newcommand{\Vv}{\mathbf{V}}
\newcommand{\HHL}{\HH^{(L)}}

\numberwithin{equation}{section}

\def\tmatrix#1#2{\left(\begin{smallmatrix} #1 \\[.5ex] #2\end{smallmatrix}\right)}

\newcommand{\tnorm}[2][]{\ifx @#1@\def\myleft{\left}%
          \def\myright{\right}%
     \else\def\myleft{\csname#1l\endcsname}%
          \def\myright{\csname#1r\endcsname}%
          \ifx\myleft\normall\def\myleft{\relax}\def\myright\relax\fi%
     \fi%
     \myleft|\mkern-2.5mu\myleft|\mkern-2.5mu\myleft| #2 \myright|\mkern-2.5mu\myright|\mkern-2.5mu\myright|}

\newcommand{\dscalar}[3][]{\ifx @#1@\def\myleft{\left}%
          \def\myright{\right}%
     \else\def\myleft{\csname#1l\endcsname}%
          \def\myright{\csname#1r\endcsname}%
          \ifx\myleft\normall\def\myleft{\relax}\def\myright{\relax}\fi%
     \fi%
     \myleft\langle\mkern\llshift\myleft\langle #2,#3 \myright\rangle\mkern\llshift\myright\rangle}

\def\llangle{\langle\mkern-4mu\langle}
\def\rrangle{\rangle\mkern-4mu\rangle}

\makeatletter

\newcounter{numcount}
\newcommand{\labelnummer}{\textup{(\roman{numcount})}}%

\newenvironment{nummer}%
{\let\curlabelspeicher\@currentlabel%
  \begin{list}{\labelnummer}{\usecounter{numcount}\leftmargin0pt%
      \topsep0.5ex\partopsep2ex\parsep0pt\itemsep0.5ex\@plus1\p@%
      \labelwidth3.5em\itemindent3.5em\labelsep1em}%
    \let\saveitem\item%
    \def\item{\saveitem%
      \def\@currentlabel{\curlabelspeicher\kern1pt\labelnummer}%
      \let\label\bemlabel}}%
  {\vskip1ex\end{list}}%

\AtBeginDocument{\let\mysaveref\ref}
\def\itemref#1{\mysaveref{item-#1}}

\AtBeginDocument{\let\yetanotherlabel\label}
\def\bemlabel#1{\yetanotherlabel{#1}
  \def\@currentlabel{\labelnummer}
  \yetanotherlabel{item-#1}}%

\def\pper{.}
\def\HarvardComma{}

\newcounter{aucount}
\newif\ifedplural
\newif\ifper\pertrue

\def\au#1#2{{#1 #2}}
\def\lau#1#2{{#1 #2}, }
\def\ed#1#2{\ifnum\theaucount=0\relax\fi{#1 #2}\addtocounter{aucount}{1}}
\def\led#1#2{\ifnum\theaucount=0\relax\edpluralfalse\else\edpluraltrue\fi{#1
    #2} (\editorname.),\setcounter{aucount}{0}}
\def\editorname{\ifedplural Eds\else Ed\fi}

\def\et{\ifnum\theaucount=1\else\HarvardComma\fi{} and\ }
\def\ti#1{\emph{#1}.\ifper\fi\pertrue}

\def\bti{\@ifnextchar[\bbti\bbbti}
\def\bbti[#1]#2{\emph{#2}, #1.}
\def\bbbti#1{\emph{#1}.}

\def\z{\@ifnextchar[\zz\zzz}
\def\zz[#1]#2#3#4#5{\perfalse{#2} \textbf{#3}, #4 \ifx
  @#5@\relax\else (#5)\fi{} [#1]\ifper\pper\fi\pertrue}
\def\zzz#1#2#3#4{{#1} \textbf{#2}, #3 \ifx @#4@\relax\else
  (#4)\fi\ifper\pper\fi\pertrue}

\def\pub{\@ifstar\pubstar\pubnostar}
\def\pubnostar{\@ifnextchar[\@@pubnostar\@pubnostar}
\def\@@pubnostar[#1]#2#3#4{#2, #3, #4, #1\ifper\pper\fi\pertrue}
\def\@pubnostar#1#2#3{#1, #2, #3\ifper\pper\fi\pertrue}
\def\pubstar[#1]#2#3#4{\perfalse #2, #3, #4 [#1]\pper\pertrue}

\def\@setauthors{%
   \begingroup
   \trivlist
   \centering\@topsep30\p@\relax
   \advance\@topsep by -\baselineskip
   \item\relax
   \andify\authors
   \def\\{\protect\linebreak}%
  \textsc{\authors}%
  \endtrivlist
   \endgroup
}

\makeatother
\title{Random block operators}

\thanks{Work supported by the German Research Foundation (DFG) through Sfb/Tr~12 ``Symmetries and Universality in Mesoscopic Systems'' and DFG Research Unit~718 ``Analysis and Stochastics in Complex Physical Systems''.}

\author{Werner Kirsch}
\address{Fakult\"at f\"ur Mathematik und Informatik, FernUniversit\"at Hagen,
L\"utzowstr.\ 125, 58084 Hagen, Germany}
\email{werner.kirsch@FernUni-Hagen.de}

\author{Bernd Metzger}
\address{Weierstrass Institute for Applied Analysis and Stochastics,
Moh\-renstr.\ 39,
10117 Berlin, Germany
}
\email{bernd.metzger@wias-berlin.de}

\author{Peter M\"uller}
\address{Mathematisches Institut der Universit\"at M\"unchen\\
	Theresienstr.\ 39\\	
	80333 M\"unchen\\
	Germany
	}
\email{mueller@lmu.de}

\date{\today}

\begin{document}

\begin{abstract}
	We study fundamental spectral properties of random block operators that are common in
	the physical modelling of mesoscopic disordered systems such as dirty
	superconductors. Our results include ergodic properties,
	the location of the spectrum, existence and regularity of the integrated density of states,
	as well as Lifshits tails. Special
	attention is paid to the peculiarities arising from the block structure
	such as the occurrence of a robust gap in the middle of the spectrum. Without
	randomness in the off-diagonal blocks the density of states typically exhibits an inverse
	square-root singularity at the edges of the gap. In the presence of
	randomness we establish a Wegner estimate that is valid at all energies. It
	implies that the singularities are smeared out by randomness, and the
	density of states is bounded. We also show Lifshits tails at these band edges. Technically, one has
	to cope with a non-monotone dependence on the random couplings.
\end{abstract}

\maketitle


\section{Introduction}


Random block operators play an important role in the mathematical modelling
of superfluid fermions in a random environment and are thus relevant for mesoscopic disordered quantum systems such as dirty superconductor devices. They arise in the Bogoliubov-de Gennes equation
\begin{equation}
	\label{BdG}
  \tmatrix{H&B}{B^{*} & - \bar{H}} \tmatrix{\psi_{+}}{\psi_{-}} = E \tmatrix{\psi_{+}}{\psi_{-}},
\end{equation}
that is, the eigenvalue problem for the quasi-particle (or excitation) states $ \binom{\psi_{+}}{\psi_{-}} $ in a mean-field approximation of BCS theory \cite{deGe66}. Without loss of generality we have assumed that the chemical potential equals zero in \eqref{BdG}. The `particle' and `hole' components $\psi_{+}$ and $\psi_{-}$ of the quasi-particle state belong to the single-particle Hilbert space $\cH$. The self-adjoint single-particle Hamiltonian $H=H^{*}$ and the so-called the pair potential or gap function $B$ are linear operators on $\cH$. The overbar in \eqref{BdG} denotes complex conjugation.

Following Altland and Zirnbauer \cite{AZ97} one can classify all block operators that arise in \eqref{BdG} according to their behaviour with respect to time-reversal and spin-rotation symmetry.
In this paper we will focus on random block operators of the form
\begin{equation}
	\label{HH-CI}
 	\HH :=  \tmatrix{H& \phantom{_{|}}B }{B & - H}
\end{equation}
with both $H$ and $B$ self-adjoint. This choice corresponds to symmetry class $C$I of \cite{AZ97} and describes physical systems for which both time-reversal and spin-rotation symmetry hold.
Since the Bogoliubov-de Gennes equation results from a mean-field approximation, the expressions for the operators $H$ and $B$ should be determined from self-consistency requirements. For disordered systems
the discrete Anderson model in $d$ dimensions,
\begin{equation}
 	H := \Delta + V \qquad\text{on}\quad \cH = \ell^{2}(\zd),
\end{equation}
is a generally accepted effective description for this, see e.g.\ \cite{ViSeFi00}. (Choosing $\HH$ as a random matrix from a suitable ensemble would be another \cite{AZ97}.) Here,
\begin{equation}
	\label{disc-laplace}
	(\Delta\psi)(j) := \sum_{i\in\zd: |i-j|=1}\psi(i)
\end{equation}
for all $j\in\zd$ and all $\psi\in\ell^{2}(\zd)$ is the centred discrete Laplacian and the random potential $V$ amounts to multiplication by independent and identically distributed, real-valued random variables
$\{V(j)\}_{j\in\zd}$ according to $(V\psi)(j) := V(j)\psi(j)$.

The form of the gap operator $B$, which should also be determined by self-consistency, depends on the pairing mechanism. For $s$-wave (a.k.a.\ conventional) superconductors $B$ is a multiplication operator in position space $\ell^{2}(\zd)$. Homogeneous $s$-wave superconductors are described by a multiple of the identity operator, $B=\beta\, 1$ with a self-consistently determined parameter $\beta >0$. Disordered $s$-wave superconductors are often described by an effective  random multiplication operator $B=b$ \cite{deGe66, Zie92, ViSeFi00}. Here $(b\psi)(j) := b(j)\psi(j)$, where $\{b(j)\}_{j\in\zd}$ are
independent and identically distributed real-valued random variables. In addition, the $b(j)$'s are often required to be independent of the $V(j)$'s.
Our main results in Sections~\ref{sec:wegner} and~\ref{sec:lif} will be proved in precisely this setting.

Non-diagonal gap operators $B$ occur in the modelling of $d_{x^2-y^2}$-wave superconductors.
For example, the momentum-dependent interaction of Cooper pairs  leads to  $B= \beta ( \Delta_x^{(1)} - \Delta_y^{(1)})$ for homogeneous superconductors in two dimensions \cite{DuLe00,ASZ02}, where $\beta >0$ and $\Delta^{(1)}_{x/y}$ denotes the one-dimensional centred discrete Laplacian in $x$-, resp.\ $y$-direction.
Our results may also be of relevance for inhomogeneous $d$-wave superconductors \cite{Zie92, ViSeFi00}, if we use models with diagonal disorder
for these materials.

The plan of this paper is as follows.
In Section~\ref{sec:block-spec} we study basic spectral features of block operators $\HH$ that are of the general form \eqref{HH-CI}. Among others we establish the existence of a robust spectral gap of $\HH$ in Proposition~\ref{spec-gap}.
We interpret the robustness of the gap in the context of Anderson's theorem \cite{And59, BaVeZh06} in Remark~\ref{gap-B-Anderson}.

Section~\ref{constndiag} briefly discusses an important special case of \eqref{HH-CI}, namely constant off-diagonal blocks $B =\beta \,1$. This serves to expose a typical phenomenon in the absence of disorder: the density of states of $\HH$ suffers from an inverse square-root singularity at the inner band edges in every dimension $d\in\NN$. The singularity is robust in the sense that it always shows up unless the density of states of $H$ vanishes at energy zero -- the location of the chemical potential.

We introduce the main objects of this paper, ergodic random block operators, in Section~\ref{sec:ergodic}. The basic spectral consequences of ergodicity are also explored there. This includes the location of the almost-sure spectrum, as well as the definition and self-averaging of the integrated density of states of $\HH$.

In Section~\ref{sec:wegner} we show that the density of states exists for suitable ergodic random block operators and that it is bounded. This is the content of Theorem~\ref{Wegner}, the main result of this paper, and follows from a Wegner estimate. Figure~\ref{Figure1} compares this situation to the one with the singularity at the inner band edges for a constant $B= \beta \, 1$ as in Section~\ref{constndiag}. Randomness in the off-diagonal blocks smooths out the singularities. We stress that the Wegner estimate of Theorem~\ref{Wegner} holds for a block random Schr\"odinger operator with a sign-indefinite single-site potential of mean zero. In contrast, for ordinary (i.e.\ non-block) random Schr\"odinger operators such a result is still missing despite a lot of recent efforts \cite{Ves10, Kru10}.

\begin{figure}
	\includegraphics[width=0.5\textwidth]{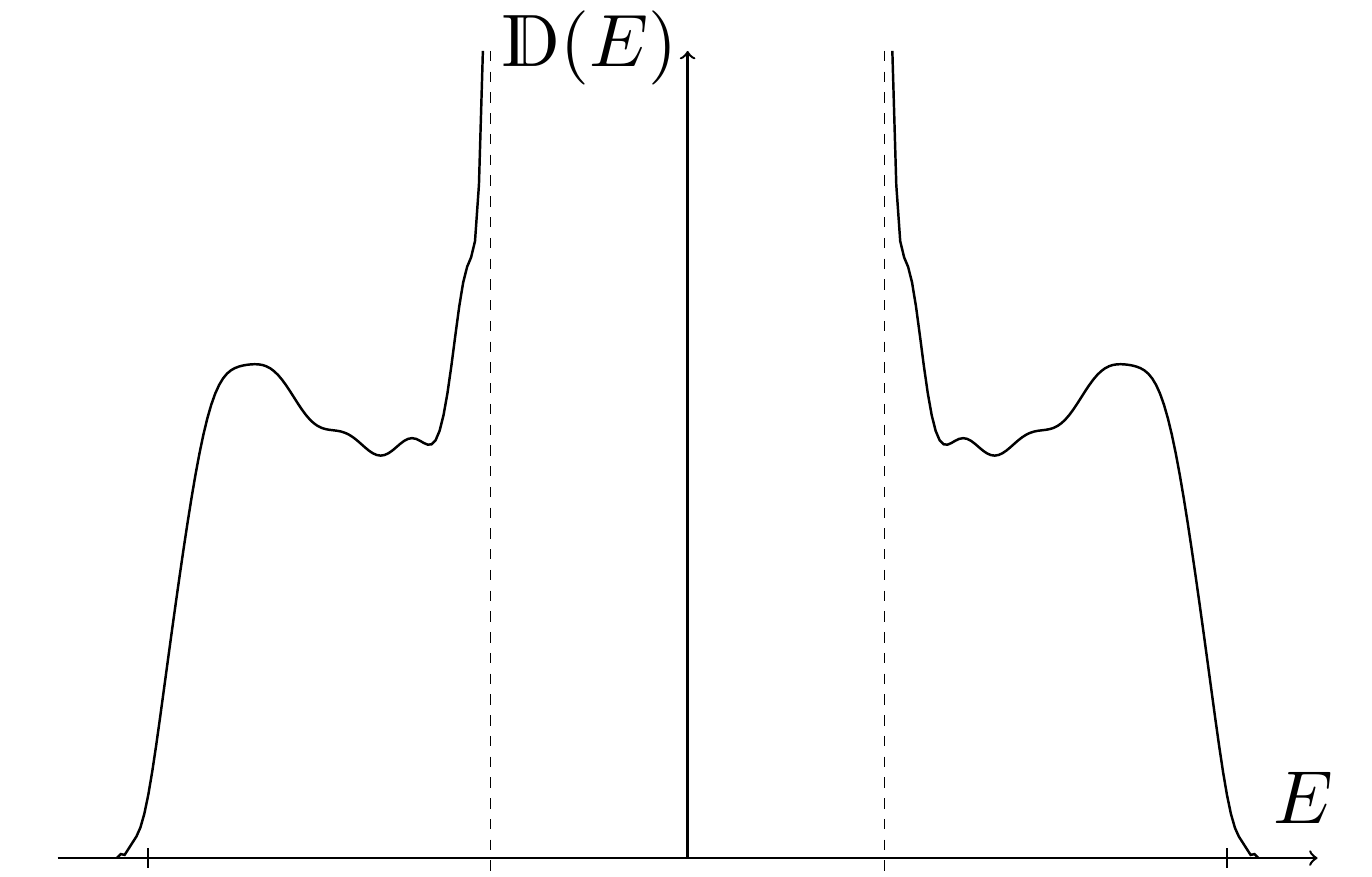}\hfill
  \includegraphics[width=0.5\textwidth]{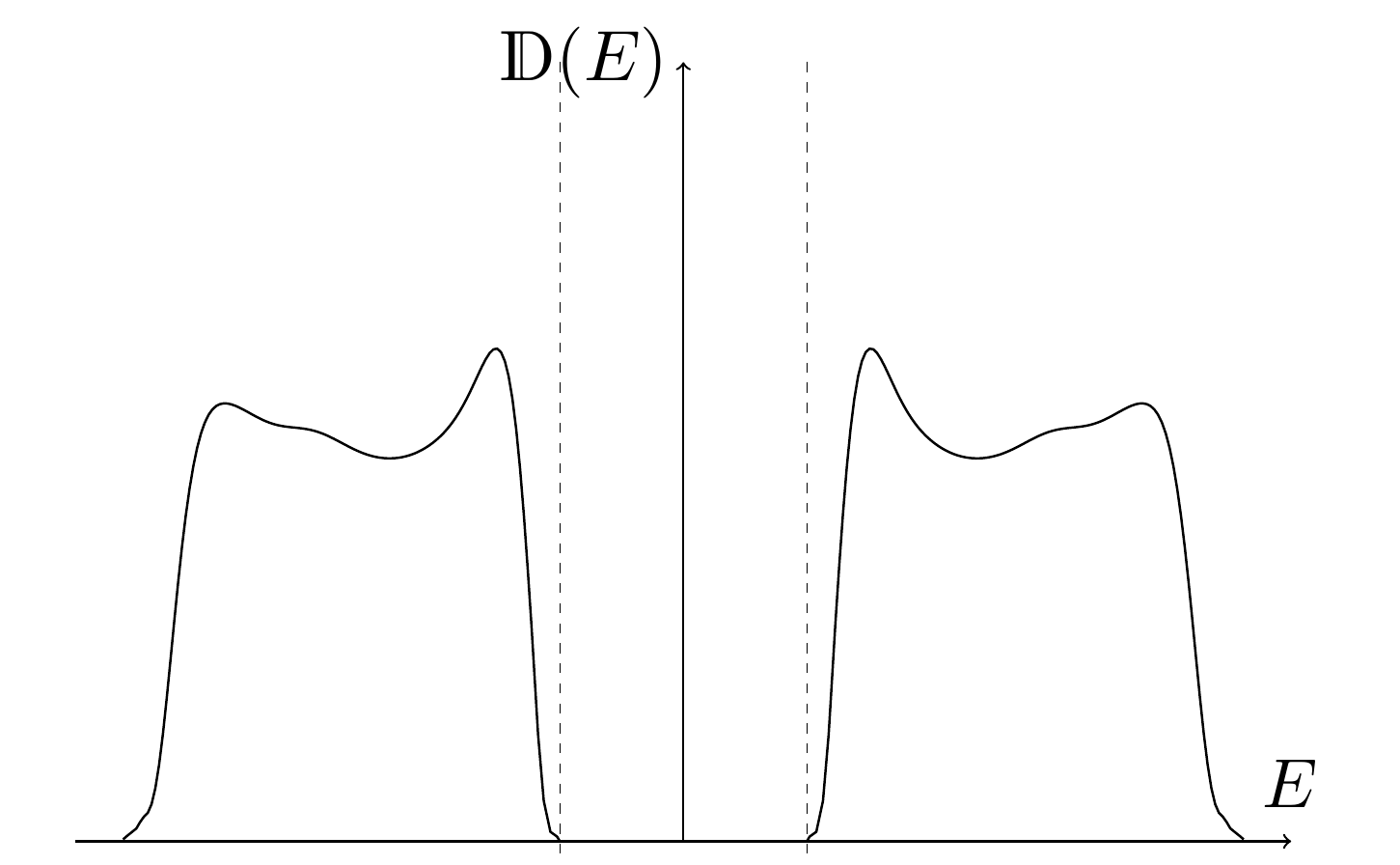}
  \caption{Comparison of the density of states of a random block operator $\HH$ without (left)
    and with (right) randomness in the off-diagonal blocks $B$. The singularities at the inner
		band edges are smoothed out by the randomness in $B$.}
  \label{Figure1}
\end{figure}

Finally, Section~\ref{sec:lif} establishes Lifshits tails for the integrated density of states  of $\HH$ at the inner band edges. Theorem~\ref{lif-gen} embodies the second main result presented here.


\section{Basic properties of block operators}
\label{sec:block-spec}


In this section we explore some fundamental properties of rather general self-adjoint block operators
\begin{equation}
\HH :=\tmatrix{H & B}{B & -H}
\end{equation}
on the Hilbert space $\cH^{2} := \cH \oplus \cH$. We equip $\cH^{2}$ with the scalar product
\begin{equation}
\llangle \Psi, \Phi\rrangle := \langle\psi_{1}, \varphi_{1}\rangle + \langle\psi_{2}, \varphi_{2}\rangle,
\end{equation}
where $\langle\cdot, \cdot\rangle$ stands for the Hilbert space scalar product
of $\cH$ and $\Psi := \binom{\psi_{1}}{\psi_{2}}, \Phi := \binom{\varphi_{1}}{\varphi_{2}} \in \cH^{2}$.
We write $\tnorm{\Psi} := (\|\psi_{1}\|^{2} + \|\psi_{2}\|^{2})^{1/2}$ for the induced norm.

Later we will be mainly interested in the case where the self-adjoint operator $H$ is a discrete Schr\"odinger operator on $\cH=\ell^2(\zd)$ and $B$ is a multiplication operator with some real-valued function $b$ on $\zd$. In this situation both operators $H$ and $B$ are frequently bounded, so that $\HH$ is unambiguously
well defined as a self-adjoint operator. However, we can treat unbounded operators as well, as one can read off from the following assertion.

\begin{proposition}
\label{selfad}
\begin{nummer}
\item Let $H,B$ be self-adjoint, assume that $\dom(B) \cap \dom(H)$ is a
	core for $H$ and that $\dom(|H|^{1/2}) \subseteq \dom(B)$. Then $\HH$ is essentially self-adjoint on $
	\big(\dom(B) \cap \dom(H)\big) \oplus \big(\dom(B) \cap \dom(H)\big)$.
\item Let $H$ be self-adjoint, let $B$ be symmetric and $H$-bounded with
	bound strictly smaller than one. Then $\HH$ is self-adjoint on $\dom(H) \oplus \dom(H)$.
\item Let $B$ be self-adjoint, let $H$ be symmetric and $B$-bounded with
	bound strictly smaller than one. Then $\HH$ is self-adjoint on $\dom(B) \oplus \dom(B)$.
\end{nummer}
\end{proposition}

\begin{proof}
 The assertions are special cases of Thm.~2.6.6 and Prop.~2.3.6 in \cite{Tre08}.
\end{proof}

Without further mentioning we will assume in the rest of this paper that at least one of the three situations described by Proposition~\ref{selfad} applies, thereby ensuring self-adjointness of $\HH$. Next we compile some basic structural properties of the spectrum of $\HH$.

\begin{lemma}\label{diag-offdiag}
	The operators $\HH$ and $\HH' := \tmatrix{B & H}{H & -B}$ are unitary equivalent.
\end{lemma}

\begin{proof}
Define the unitary involution $\UU_1 := \frac{1}{\sqrt{2}} \tmatrix{1 & 1}{1 & -1}$, then $\HH'=\UU_1 \,\HH\, {\UU_1}^*$.
\end{proof}

\begin{lemma}
	\label{symm-spec}
	The spectrum of $\HH$ is symmetric around $0$, i.e.
	\begin{equation}
 	\spec(\HH) = - \spec (\HH).
	\end{equation}
	In particular, if $\HH\Psi = E \Psi$ for some $E\in\RR$ and $\Psi = (\psi_{1},\psi_{2})^T \in \cH^{2}$, then $\HH \tilde{\Psi} = -E \tilde{\Psi}$, where $\tilde{\Psi} = (\psi_{2},-\psi_{1})^T$.	 
\end{lemma}

\begin{proof}
 We define the unitary transformation $\UU_{2} :=
\tmatrix{0 & 1}{-1 & 0}$ on $\cH^{2}$, which obeys $\UU_{2}^{-2} = - \UU_{2}$, and observe
$\UU_{2} \HH \UU_{2}^{*} = - \HH$.
\end{proof}

\begin{remarks}
\item It follows that the spectrum of $\HH^{2}$ has multiplicity at least 2, except possibly at zero, and that 	
\begin{equation}
 \spec(\HH)=\{E\in\RR \mid E^2\in\spec(\HH^2)\}.
\end{equation}
\item The anti-symmetry of $\HH$ under the transformation $\UU_2$ is known as `particle-hole' symmetry.
\end{remarks}

\begin{lemma} \label{square}
The operator $\HH^2$ is given by
\begin{equation}
	\HH^{2} = \tmatrix{H^{2} + B^{2} & \;[H,B]}{ -[H,B] & \;H^{2} + B^{2}}
\end{equation}
and unitarily equivalent to $K_{-} \oplus K_{+}$ on $\cH^{2}$, where
\begin{equation}
	K_{\pm} := H^{2} + B^{2}  \pm \i [H,B].
\end{equation}
\end{lemma}

\begin{proof} This follows from an explicit computation and the observation that
\begin{align}
 \UU_{3} \HH^{2} \UU_{3}^{*} = \tmatrix{K_{-} & 0}{0 & K_{+}}
\end{align}
for the unitary $\UU_{3} := \frac{1}{\sqrt{2}} \tmatrix{1 & \i}{\i & 1}$.
\end{proof}

\begin{remark}
\label{non-neg}
 Introducing the annihilation operator $a := H -\i B$ on $\cH$, we can write $K_{+} = a a^{*}$ and
 $K_{-} = a^{*} a$. Thus, the spectra of $K_{+}$ and $K_{-}$ differ at most by $\{0\}$. In fact, we have
\begin{equation}
 \tmatrix{K_{-} & 0}{0 & K_{+}} = \tmatrix{0 & a^{*}}{a & 0}^2
 \qquad\text{and}\qquad \UU_{3}\HH\UU_{3}^{*} = \i \tmatrix{0 & -a^{*}}{a & 0} .
\end{equation}
\end{remark}

A direct calculation also shows the next assertion.

\begin{lemma}
	\label{gen-unitary}
 Suppose there exists a unitary involution $U = U^{*} = U^{-1}$ on $\cH$ such that $HU + UH = 0$ and $[B,U] =0$. Then $\UU  := \frac{1}{\sqrt{2}}  \tmatrix{1 & U}{1 & -U}$ is unitary on $\cH^{2}$ and
$\UU \HH \UU^{*}= \tmatrix{H+UB &0}{0 & H - UB}$.
\end{lemma}

Later we will use the following particularisation of Lemma~\ref{gen-unitary}.

\begin{corollary}
\label{unitary} Assume $B\equiv b$ is the
maximal self-adjoint multiplication operator by the function $b: \zd \rightarrow \RR$ and $\Delta$ the centred discrete Laplacian as defined in \eqref{disc-laplace}. Then the operator $\HH =
\tmatrix{\Delta & b}{b & -\Delta}$ is unitarily equivalent to $H_{+} \oplus H_{-}$, where
 $H_{\pm} := \Delta \pm (-1)^{j}b$.
\end{corollary}

\begin{proof}
Choose  $U:= (-1)^{j}$, that is
\begin{equation}
(U\psi)(j) := (-1)^{j} \psi(j) := (-1)^{\sum_{k=1}^{d}j_{k}}\psi(j)
\end{equation}
for all $j=(j_{1},\ldots,j_{d}) \in \zd$ and
$\psi\in\ell^{2}(\zd)$, and apply Lemma~\ref{gen-unitary}.
\end{proof}

\begin{remarks}
\item It is essential for the validity of Corollary~\ref{unitary} that $\Delta$ contains
    no diagonal terms.
\item Lemma~\ref{square} and Corollary~\ref{unitary} imply that $\HH^{2}$ is both
    unitarily equivalent to $K_{-} \oplus K_{+}$ and also to $H_{+}^{2} \oplus
    H_{-}^{2}$. We have $U K_{+} U = K_{-}$ and $UH_{+}U = -H_{-}$. Note also that none
    of the operators $H_{+}^{2}$ or $H_{-}^{2}$ alone is unitarily equivalent to $K_{+}$
    or $K_{-}$.
\item From Corollary \ref{unitary} we can infer some information on the location of the spectrum of discrete Schr\"odinger operators with periodic potentials. For example, if $b$ is the constant function equal $\beta\in\RR$, then the corresponding operators $H_\pm$
are unitarily equivalent. Hence $\spec(H_+)=\spec(H_-)=\spec(\HH)$. Thus it will follow from
Proposition~\ref{spec-gap} below, that the interval $]-\beta,\beta[$ is a spectral gap for the operator $\Delta +(-1)^j\,\beta$.
\end{remarks}

Finally, we are concerned with locating the spectrum of general $\HH$
and with the occurrence of spectral gaps that arise from the special block structure.

\begin{proposition} \label{spec-gap}
	\begin{nummer}
	\item \label{spec-include}
	$\spec(\HH) \subseteq \big[ - \|H\| - \|B\|, \|H\| + \|B\| \big]$.
	\item \label{gap-H}
		Suppose there exists $ \lambda \ge 0$ such that $H \ge \lambda 1$. Then
		\begin{equation}
 			\spec (\HH) \;\cap\; ]-\lambda, \lambda [ = \emptyset.
		\end{equation}
	\item \label{gap-B}
		Suppose there exists $\beta \ge 0$ such that $B \ge \beta 1$. Then
		\begin{equation}
 			\spec (\HH) \;\cap\; ]-\beta, \beta [ = \emptyset.
		\end{equation}
	\item \label{gap}
		Suppose there exist $\lambda, \beta \ge 0$ such that  $H \ge \lambda 1$ and $B \ge \beta 1$. Then
		\begin{equation}
 			\spec (\HH) \;\cap\; \big]- (\lambda^{2}+\beta^{2})^{1/2}, (\lambda^{2}+\beta^{2})^{1/2} \big[ = \emptyset.
		\end{equation}
	\end{nummer}
\end{proposition}

\begin{remarks}
\item \label{end-sharp}
	The endpoints of the interval in Proposition \ref{spec-include} are sharp upper and lower
	bounds for the maximum and minimum of $\spec(\HH)$. This can be seen by choosing $B$ as the
	multiplication operator by the function $b=(-1)^{j}$ in Corollary~\ref{unitary}.
\item The statements of Proposition~\ref{gap-H}, \itemref{gap-B} and~\itemref{gap} remain true if one
 replaces $H$ and $B$ by $-H$ and $-B$ in the assumptions.
\item \label{gap-B-Anderson}
	According to Proposition \ref{gap-B}, $\HH$ has always a spectral gap of size at least $2\beta$, no matter what $H$ is like. This statement can be interpreted in the context of Anderson's theorem \cite{And59, BaVeZh06}: Anderson argued that	adding non-magnetic impurities to an $s$-wave superconductor should have only little (or at best no) effect on the gap operator $B$. In view of Proposition~\ref{gap-B}, this implies stability of the gap under non-magnetic impurity doping and thus leads to the experimentally observed insensitivity of superconductivity in this case.
As an aside we mention that the situation is totally different when doping $d_{x^2-y^2}$-wave superconductors by non-magnetic impurities. Here it is known that disorder leads to the breaking of Cooper pairs.
\end{remarks}

\begin{proof}[Proof of Proposition~\ref{spec-gap}]
The statement in \itemref{spec-include} follows from Lemma \ref{square} and
$\|K_{\pm}\| \le ( \|H\| + \|B\|)^{2}$.

To prove (iv) we define $\tilde{H} := H - \lambda \ge 0$ and $\tilde{B} := B -\beta \ge 0$.
Then we get for the operators $K_{\pm}$ in Lemma~\ref{square}
\begin{equation}
 K_{\pm} = \tilde{H}^{2} + \tilde{B}^{2} \pm \i [\tilde{H},\tilde{B}]
 + \lambda^{2} + 2\lambda\tilde{H} + \beta^{2} + 2\beta\tilde{B} \ge \lambda^{2} + \beta^{2},
\end{equation}
where we used that the linear terms in $\lambda$ and $\beta$ are manifestly non-negative,
as is the sum of the first three terms on account of Remark~\ref{non-neg}. Hence, $\HH^{2} \ge \lambda^{2} + \beta^{2}$.

The proofs of~\itemref{gap-H} and~\itemref{gap-B} are simpler and proceed along the very same lines.
\end{proof}

Part~\itemref{gap-H} of Proposition~\ref{spec-gap} allows for a generalisation towards different diagonal blocks.

\begin{lemma}
 \label{gap-different}
 Let $H_{1}$ be a self-adjoint operator on $\cH$ such that the block operator $\HH_{1} := \tmatrix{H & B}{B & -H_{1}}$ is self-adjoint on $\cH^{2}$. Suppose there exists $ \lambda \ge 0$ such that both $H \ge \lambda $ and $H_{1} \ge \lambda$. Then
		\begin{equation}
 			\spec (\HH_{1}) \;\cap\; ]-\lambda, \lambda [ = \emptyset.
		\end{equation}
\end{lemma}

\begin{proof}
 We define $\tilde{H} := H - \lambda \ge 0$ and $\tilde{H}_{1} := H_{1} - \lambda \ge 0$. An explicit computation gives
\begin{align}
 \HH_{1}^{2} &= \tmatrix{H^{2} + B^{2} & HB-BH_{1}}{BH-H_{1}B & H_{1}^{2} + B^{2}}
 = \tmatrix{\lambda^{2} + 2 \lambda \tilde{H} & 0}{ 0 & \lambda^{2} + 2 \lambda \tilde{H}_{1}}
 + \tmatrix{\tilde{H} & B}{B & -\tilde{H}_{1}}^{2} \nonumber\\[.5ex]
 &\ge \lambda^{2} \one.
\end{align}
\end{proof}


\section{The case of constant diagonal $B$}  \label{constndiag}


As a warm-up and in order to expose some typical features
we consider first the simple case where $B=\beta\,1$ is a constant multiple of the
identity. In particular, the operators $H$ and $B$ commute. As a consequence, the spectral
theory of $\HH$ can be reduced to the diagonalisation of a $2\times 2$-matrix and the
spectral theory of $H$.

\begin{proposition}\label{constantndiag-spec}
	Let $H$ be a self-adjoint operator on $\cH$, let $\beta \in \RR\setminus \{0\}$ and consider
	\begin{equation}
		\label{const-op}
		\HH :=\sma{H}{\beta 1}
	\end{equation}
	on $\cH^{2}$. Then
	\begin{nummer}
 	\item  \label{const-spec}
			$\spec(\HH) = \big\{ E_{\pm} := \pm (E^{2} + \beta^{2})^{1/2} :
			E \in \spec(H) \big\}$.
	\item \label{const-ef}
		If $H\varphi = E\varphi$ for some $E\in \spec(H)$ and $\varphi\in\cH$, then
		$\HH \tilde\Phi_{\pm} = E_{\pm} \tilde\Phi_{\pm}$
		for the non-normalised vector
		$\tilde\Phi_{\pm} := \binom{\varphi}{\beta^{-1}(E_{\pm} - E)\varphi} \in \cH^{2}$.
	\item \label{const-weyl}
	 Fix also $E \in \RR$. Then there exists a constant $C \equiv C_{E,\beta} >0$
		such that if $\| (H-E)\varphi\| \le \varepsilon$ for some $\varphi\in\cH$ and $\varepsilon>0$, then
		\begin{equation}
			\tnorm[normal]{(\HH -E_{\pm}) \tilde\Phi_{\pm}} \le C \varepsilon.
		\end{equation}
	\end{nummer}
\end{proposition}

\begin{proof}
Part~\itemref{const-spec} is a corollary of Lemma~\ref{square}, Part~\itemref{const-ef} follows from an explicit computation. As to Part~\itemref{const-weyl} we observe
\begin{equation}
 	\HH \tilde\Phi_{\pm} = \tmatrix{H\varphi + (E_{\pm} -E)\varphi}{\beta\varphi -
	\beta^{-1}(E_{\pm} -E) H\varphi}
	= E_{\pm}\tilde\Phi_{\pm} + \tmatrix{(H-E)\varphi}{-\beta^{-1}(E_{\pm} -E)(H-E)\varphi}.
\end{equation}
This implies
\begin{equation}
 	\tnorm[normal]{(\HH -E_{\pm})\tilde\Phi_{\pm}}^{2} \le \varepsilon^{2} + \beta^{-2}(E_{\pm} -E)^{2}
		\varepsilon^{2} =: C^{2} \varepsilon^{2}.
\end{equation}
\end{proof}

\begin{remark}
	We know already from Proposition~\ref{gap-B} that the spectrum of $\HH$ as given by \eqref{const-op}
	has always a gap between $-|\beta|$ and $|\beta|$ no matter how the spectrum of $H$ looks like. For an interpretation of this in the context of Anderson's theorem, see Remark~\ref{gap-B-Anderson}.
\end{remark}

We can also compute the density of states of $\HH=\binom{H \;\; \beta 1}{\beta 1 \; -H}$, given the operator $H$ possesses a density of states. The corresponding result is stated and proven here, even though we postpone the formal definition of the integrated density of states and of its Lebesgue derivative, the density of states, to the next section.

\begin{figure}
	\includegraphics[width=0.5\textwidth]{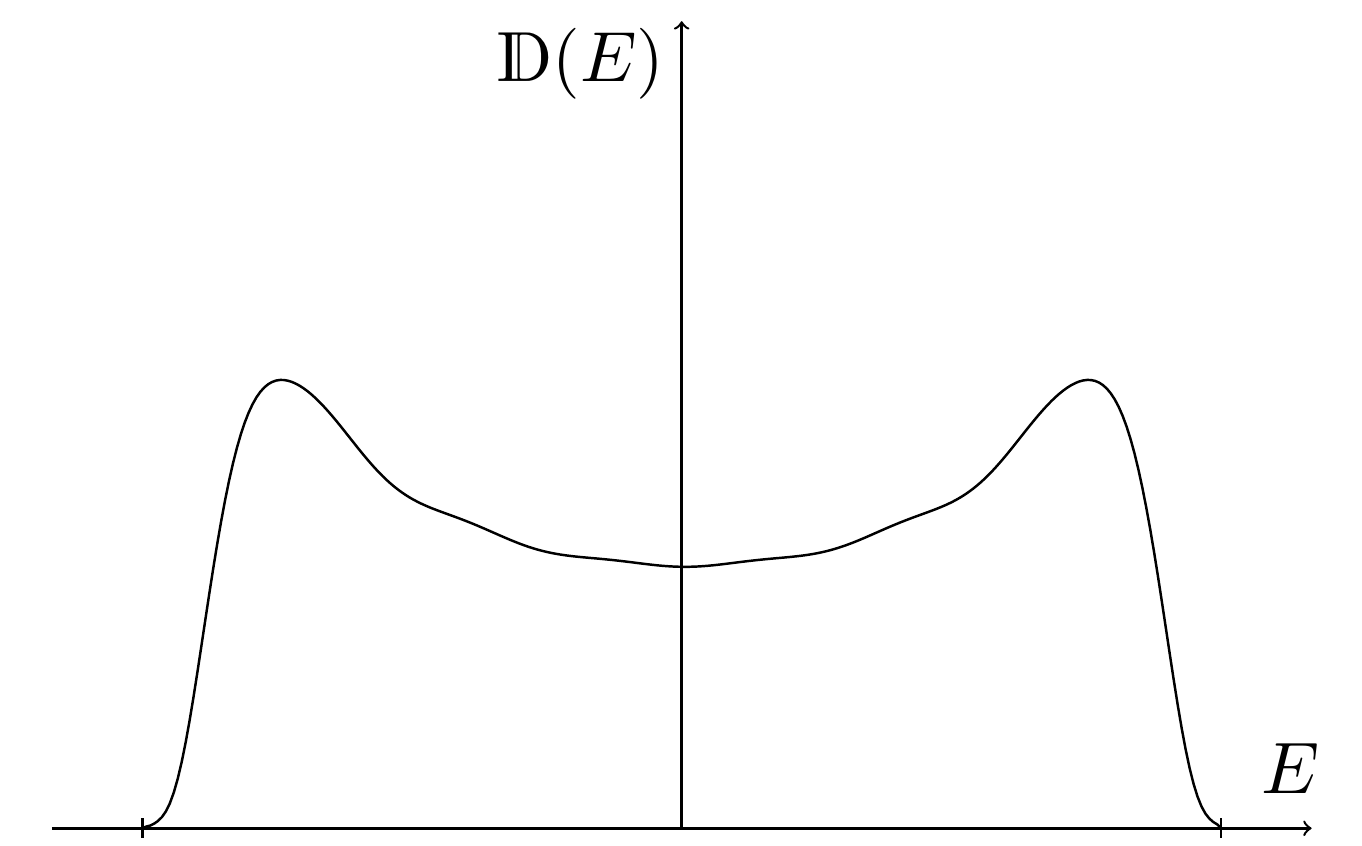}\hfill
  \includegraphics[width=0.5\textwidth]{intro-w-aus.pdf}
  \caption{Comparison of the density of states of a random block operator $\HH$ with $B=0$ (left)
    and $B= \beta 1$ (right).}
  \label{Figure2}
\end{figure}

\begin{proposition}
\label{const-dos}
If $H$ possesses an absolutely continuous integrated density of states with
Lebesgue derivative $D$, then the integrated density of states of $\HH$ is also absolutely
continuous with Lebesgue derivative given by
\begin{equation}
 	\DD(E) = \frac{|E|}{\sqrt{E^{2} - \beta^{2}}}\left[ D\big( \textstyle\sqrt{E^{2} - \beta^{2}}\,  \big)
	+ D\big( - \sqrt{E^{2} - \beta^{2}} \,\big) \right]
\end{equation}
for a.e.\ $E \in\spec(\HH) \subseteq \RR \setminus ]-|\beta|, |\beta|[$, and $\DD(E)
= 0$ for $E \notin\spec(\HH)$.
\end{proposition}

\begin{remark}
Observe that in the case considered here the density of states $\DD$ of $\HH$ has a
square-root singularity at the inner band edges $\pm\beta$ whenever $D(0+)\not=0$, see the right panel of Figure~\ref{Figure2}. This is typically the case, if $0$ lies in the interior of the spectrum of a random Schr\"odinger operator $H$ \cite{We91,Je92,HisMu08}. The square-root singularity may remind the reader of the van Hove
singularity in dimension $d=1$. In the present case the nature of the singularity is
independent of the dimension $d$, however.
\end{remark}

\begin{proof}
By assumption we have for every interval $A_{0} \subseteq \RR$
\begin{equation}
 	\int_{A_{0}}\!\d E_{0}\; D(E_{0}) = \lim_{L \to\infty}
 	\frac{\# \big\{\text{eigenvalues of $H^{(L)}$ lying in $A_{0}$}\big\}}{|\Lambda_L|},
\end{equation}
where $H^{(L)}$ is a self-adjoint finite-volume restriction of $H$ to the cube $\Lambda_L$ centred about the origin and containing $|\Lambda_{L}| =L^{d}$ points for $L \in\NN$ odd. In particular, the spectrum of $H^{(L)}$ is discrete and $H^{(L)}$ converge to $H$ in the limit $L\to \infty$.
Now, for a given interval $A \subseteq [0,\infty[$ we define
\begin{equation}
	A_{0} := \big\{E_{0} \in\RR : \sqrt{\smash{E_{0}^{2}} + \beta^{2}} \in A \big\}
\end{equation}
and
\begin{equation}
	\HH^{(L)} := \tmatrix{H^{(L)} & \beta 1}{\beta 1 & -H^{(L)}}.
\end{equation}
Proposition~\ref{constantndiag-spec} implies
\begin{multline}
	\# \big\{\text{eigenvalues of $H^{(L)}$ lying in $A_{0}$} \big\} \\
	=
	\# \big\{\text{eigenvalues of $\HH^{(L)}$ lying in $A$}\big\}.
\end{multline}
On the other hand, we have the identity
\begin{equation}
 	\int_{A_{0}}\!\d E_{0}\; D(E_{0}) = \int_{A}\!\d E\; \DD(E),
\end{equation}
which results from the change-of-variables $E_{0} =  \sqrt{E^{2} - \beta^{2}}$, $E>0$, on
$A_{0} \cap [0,\infty[$ and $E_{0} =  - \sqrt{E^{2} - \beta^{2}}$, $E>0$, on $A_{0} \cap ]-\infty, 0]$. Altogether we have
\begin{equation}
	\label{dosHH}
  \int_{A}\!\d E\; \DD(E) = \lim_{L \to\infty} \frac{\# \big\{\text{eigenvalues of $\HH^{(L)}$
  lying in $A$}\big\}}{|\Lambda_L|}.
\end{equation}
The same argument applies if $A \subseteq ]-\infty, 0[$. Thus, \eqref{dosHH} generalises
to arbitrary Borel sets $A \subseteq \RR$, proving the assertion.
\end{proof}



\section{Ergodic properties of random block operators}\label{sec:ergodic}


>From now on we are concerned with random block operators. We consider a probability space $(\Omega,\mathcal{F}, \PP)$ together with a block-operator-valued random variable
\begin{equation}
 \label{ran-HH-gen}
 \HH: \omega \mapsto \HH_{\omega} := \tmatrix{H_{\omega} & B_{\omega}}{B_{\omega} & -H_{\omega}},
\end{equation}
where $\HH_{\omega}$ is densely defined on $\cH^{2} = \ell^{2}(\zd) \oplus\ell^{2}(\zd)$ for
$\PP$-a.e.\ $\omega\in\Omega$.
Throughout we will assume that (at least) one of the situations of Proposition~\ref{selfad}
applies to $\HH_{\omega}$ for $\PP$-a.e.\ $\omega\in\Omega$ so that self-adjointness is ensured.

We say that the random block operator $\HH$ is \emph{ergodic} (w.r.t.\ $p\zd$-translations),
if there exists a period $p\in\NN^{d}$ and an ergodic group of measure-preserving
transformations $\{\tau_{j}\}_{j\in p\zd}$ on $\Omega$ such that $\HH$ fulfils the covariance
relation $\UU_{j}\HH_{\omega}\UU_{j}^{*} = \HH_{\tau_{j}(\omega)}$ for every $\omega\in\Omega$
and for every $j\in p\zd = \timesop_{k=1}^{d} (p_{k}\ZZ)$, where
$\UU_{j} := \tmatrix{U_{j} & 0}{0 & U_{j}}$ and $U_{j}$ is the unitary translation operator on
$\ell^{2}(\zd)$, that is $U_{j} \varphi := \varphi(\cdot -j)$ for every $\varphi \in \ell^{2}(\zd)$.

Consequently,  standard results \cite{Kir89, CaLa90, PaFi92, Kir08} imply the existence of a non-random
closed set $\Sigma \subseteq\RR$, the \emph{a.s.\ spectrum} of $\HH$, such that
$\spec(\HH_{\omega}) =\Sigma$ for $\PP$-a.e.\ $\omega\in\Omega$.
Analogous statements hold for the components in the Lebesgue decomposition of the spectrum.

Next we introduce the central quantity of this paper which measures the density of spectral values of
$p\zd$-ergodic random block operators $\HH$.
\begin{definition}
The (non-random) right-continuous, non-decreasing function $\IZ : \RR \rightarrow [0,1]$, defined by 	
\begin{equation}
\label{IDS-def-eq}
\IZ(E) :=  \frac{1}{2\,|\Lambda_0|} \; \;\EE \big[\tr_{\cH^{2}} [ \Chi_{\Lambda_0} \Chi_{]-\infty, E]}(\HH) ] \big]
\end{equation}
for all $E\in\RR$, is called \emph{integrated density of states} of $\HH$.
Here we introduced the elementary cell
$\Lambda_0 := \{j\in\zd : 0 \le j_{k} < p_{k} \;\text{for all}\; k=1,\ldots, d\}$, $\EE$
stands for the probabilistic expectation on $\Omega$, the trace extends over $\cH^{2}$
and the notation $\Chi_{\Lambda_0}$ refers to the multiplication operator
$\tmatrix{M & 0}{0 & M}$ where $M$ is multiplication by  the indicator function of $\Lambda_0$
on $\ell^{2}(\zd)$.
\end{definition}

As we shall see from the next lemma, the definition of $\IZ$ is justified by
Birkhoff's ergodic theorem. For $L\in\NN$ odd, we denote by $\Lambda_{L}(j)$ the cube centred
about $j\in\zd$ and containing  $|\Lambda_{L}(j)| =L^{d}$ many points of $\zd$. We also write $\Lambda_{L} := \Lambda_{L}(0)$.

\begin{lemma}
Let $\HH$ be the random, $\PP$-a.s.\ self-adjoint block matrix operator \eqref{ran-HH-gen},
which is ergodic w.r.t.\ $p\zd$-translations. Then there exists a set $\Omega_{0} \subseteq \Omega$
of full probability, $\PP(\Omega_{0})=1$, such that
\begin{equation}
  \IZ(E) = \lim_{L\to\infty} \frac{1}{2\,L^{d}} \; \tr_{\cH^{2}}
  \big[ \Chi_{\Lambda_{L}(j)} \Chi_{]-\infty, E]}(\HH_{\omega}) \big]
\end{equation}
for every $E\in\RR$, every $\omega\in\Omega_{0}$ and every $j\in\zd$.
Moreover, the set of growth points of $\IZ$ is given by the a.s.\ spectrum of $\HH$.
\end{lemma}

\begin{proof}
 This is fully analogous to \cite{KiMe07} or Sect.~5.1 in \cite{Kir08}.
\end{proof}

Mostly we will be interested in more specific random block operators whose diagonal blocks are given
by the discrete random
Schr\"odinger operator of the Anderson model in $d$ dimensions and whose off-diagonal blocks
are i.i.d.\ random multiplication operators. For simplicity and ease of presentation we dispense with the presence of magnetic fields. More precisely, let $U_{0} : \zd \rightarrow\RR$ be a periodic potential with period $p \in \NN^{d}$ and define the bounded periodic background Hamiltonian $H_{0} := -\Delta + U_{0}$ on $\cH= \ell^{2}(\zd)$.
Consider two independent sequences of independent real-valued
random variables $\big(V(j)\big)_{j\in\zd}$ and
$\big(b(j)\big)_{j\in\zd}$ on $\Omega$ such that all the $V(j): \omega \mapsto V_{\omega}(j)$ 's are
identically distributed with law $\mu_{V}$ and all the $b(j): \omega\mapsto b_{\omega}(j)$ 's are identically distributed with
law $\mu_{b}$, and such that the $V$'s are independent of the $b$'s. For the sake of technical
simplicity we assume that the supports of the measures $\mu_{V}$ and $\mu_{b}$ are bounded subsets of $\RR$.
For each elementary event $\omega\in\Omega$ the functions $V_{\omega}: j \mapsto V_{\omega}(j)$ and
$b_{\omega}: j \mapsto b_{\omega}(j)$ give rise to corresponding multiplication operators.  Then the random operators
\begin{equation}
	\label{ran-Schroed}
 H: \omega\mapsto H_{\omega} :=  H_{0} + V_{\omega}
\end{equation}
and
\begin{equation} \label{ran-offdiag}
 b: \omega \mapsto b_{\omega}
\end{equation}
are $\PP$-a.s.\ bounded and self-adjoint on $\ell^{2}(\zd)$.
They are both ergodic w.r.t.\ $p\zd$-translations. From standard results, see e.g.\ \cite{Kir89, CaLa90, PaFi92, Kir08}, we know that their spectra are $\PP$-almost surely given by
\begin{equation}
\label{entry-spec}
\spec(H) = \spec(H_{0}) + \supp(\mu_{V})
\end{equation}
and
\begin{equation}
\spec(b) = \supp(\mu_{b}) .
\end{equation}
The corresponding random block operator
\begin{equation}
 \label{ran-HH-diag}
 \HH: \omega \mapsto \HH_{\omega} := \tmatrix{H_{\omega} & b_{\omega}}{b_{\omega} & -H_{\omega}},
\end{equation}
is also ergodic and $\PP$-a.s.\ bounded and self-adjoint on $\cH^{2}$. The next lemma roughly locates its a.s.\ spectrum.

\begin{lemma}
	\label{loc-spec}
  Let $\HH$ be the random block operator \eqref{ran-HH-diag}. Then
  we have $\PP$-a.s. the inclusions
\begin{equation}
\label{loc-spec-eq}
   \big\{ \pm \textstyle\sqrt{E^{2}  + \beta^{2}} : E \in \spec(H), \beta \in \supp(\mu_{b})\big\}
   \subseteq \spec(\HH) \subseteq[-r,r]
\end{equation}
where $r:= \sup_{E \in \spec(H)}  |E|  + \sup_{\beta \in \supp(\mu_{b})} |\beta|$.
\end{lemma}

\begin{remark}
Since $\sqrt{E^{2} + \beta^{2}} < E + \beta$, if both $E,\beta >0$, we suspect from recalling Remark~\ref{end-sharp} that even the left inclusion in \eqref{loc-spec-eq} might be a strict one in certain cases. This is indeed true, as can be seen by choosing $H = -\Delta$ to be deterministic and $\mu_{b}$ a symmetrical distribution w.r.t.\ reflection at zero. In this situation we may apply Corollary~\ref{unitary} and get $\PP$-a.s.
\begin{equation}
  \spec(\HH) = \bigcup_{\kappa = \pm} \spec(H_{\kappa}) = \spec(\Delta) + \supp (\mu_{b}).
\end{equation}
\end{remark}

Combining Lemma~\ref{loc-spec} with Lemma~\ref{gap} we infer

\begin{corollary}
 	\label{gap-end}
	Assume that $\inf\spec(H) \ge0$, $\inf\supp(\mu_{b}) \ge 0$ 	and
	\begin{equation}
 		\lambda_{\pm} := \pm \sqrt{[\inf\spec(H)]^{2} +
		[\inf\supp(\mu_{b})]^{2}
		} >0.
	\end{equation}
	Then $\lambda_{-}$ and $\lambda_{+}$ are the endpoints of the open gap interval which separates the positive and negative parts of the a.s.\ spectrum of $\HH$.
\end{corollary}

\begin{proof}[Proof of Lemma~\ref{loc-spec}]
While the right inclusion follows immediately from Lemma~\ref{spec-include}	and \eqref{entry-spec}, the left inclusion is based on a typical Weyl-sequence argument.
 Let $(\varphi_{n})_{n\in\NN}$ be a Weyl sequence of normalised vectors for $H_{0}$ and some fixed energy $E_{0} \in \spec(H_{0})$. Since $H_{0}$ is a bounded operator we may assume without loss of generality that each $\varphi_{n} \neq 0$ has compact support.
 The Borel-Cantelli lemma implies the existence of a set $\Omega_{0}$ of full probability, $\PP(\Omega_{0}) =1$, such that for every $\omega\in\Omega_{0}$, for every length $L>0$, for every $\varepsilon >0$, for every $v \in \supp \mu_{V}$ and every $\beta\in\supp \mu_{b}$ there exist $k \in\zd$ with
\begin{equation}
\label{approx}
 	|V_{\omega}(j) -v | < \varepsilon \quad \text{and} \quad |b_{\omega}(j) -\beta| < \varepsilon \quad \text{for all~} j\in \Lambda_{L}(k).
\end{equation}
\noindent
Now fix also $v \in \supp \mu_{V}$ and $\beta\in\supp \mu_{b}$. We define
\begin{align}
E_{\pm} := \pm \sqrt{(E_{0}+ v)^{2} + \beta^{2}},
\end{align}
and, if $\beta\neq 0$, the vector
\begin{align}
\Phi_{n}^{\pm} := \mathcal{N}_{n}\, \tmatrix{\varphi_{n}}{\beta^{-1}(E_{\pm} -E_{0}- v)\varphi_{n}},
\end{align}
where the constant $\mathcal{N}_{n}$ ensures proper normalisation $\tnorm[normal]{\Phi_{n}^{\pm}} =1$. If $\beta =0$, we set $\Phi_{n}^{+} := \binom{\varphi_{n}}{0}$ and $\Phi_{n}^{-} := \binom{0}{\varphi_{n}}$.
For every $n\in\NN$ we have the estimate
\begin{align}
 &\tnorm[normal]{(\HH_{\omega}  - E_{\pm}) \Phi_{n}^{\pm}} \nonumber \\
 &\qquad\le\phantom{:}
 \tnorm{\textstyle \tmatrix{V_{\omega} -v & b_{\omega} -\beta}{b_{\omega}-\beta & \,-(V_{\omega} -v)} \Phi_{n}^{\pm}} + \tnorm{ \, \left[\tmatrix{H_{0} +v & \beta}{\beta & \,-(H_{0} +v)}  - E_{\pm} \right]\Phi_{n}^{\pm}}
 \nonumber\\
 & \qquad = : \;T_{n}^{(1)} + T_{n}^{(2)}
\end{align}
and the identity
\begin{align}
\big(T_{n}^{(1)}\big)^{2}   = \dscalar{\Phi_{n}^{\pm}}{ \tmatrix{ (V_{\omega}-v)^{2} + (b_{\omega}- \beta)^{2}& 0 }{ 0  &(V_{\omega}-v)^{2} + (b_{\omega}- \beta)^{2}} \Phi_{n}^{\pm}}.
\end{align}
For every $n\in\NN$ and every $\varepsilon>0$ we exploit the freedom to shift the support of $\varphi_{n}$ such that it lies inside some large enough box $\Lambda_{L}(k)$ for which \eqref{approx} holds. This implies
\begin{align}
 \big(T_{n}^{(1)}\big)^{2} \le 2 \varepsilon^{2} \llangle \Phi_{n}^{\pm},\Phi_{n}^{\pm}\rrangle  = 2\varepsilon^{2}.
\end{align}
On the other hand, given any $\varepsilon >0$, we infer from the Weyl-sequence property of $(\varphi_{n})_{n\in\NN}$ for $H_{0}$ and Proposition~\ref{constantndiag-spec} that there exists $n_{\varepsilon}\in\NN$ such that for all $n \ge n_{\varepsilon}$ the estimate $T_{n}^{(2)} \le C \varepsilon$ holds with some constant $C>0$ depending on $E_{0}, v$ and $\beta$. This proves $\tnorm[normal]{(\HH_{\omega} - E_{\pm}) \Phi_{n}^{\pm}} \le (C+\sqrt{2}) \varepsilon$, and hence $(\Phi_{n}^{\pm})_{n\in\NN}$ is a Weyl sequence for $\HH$ and $E_{\pm}$.
\end{proof}

In order to make manifest the interpretation of the integrated density of states $\IZ$ as an eigenvalue counting function we have to introduce appropriate finite-volume restrictions of the block operator $\HH$.

\begin{definition}
  For $\HH$ as in \eqref{ran-HH-diag} and a finite cube $\Lambda_{L} \subset \ZZ^{d}$ we introduce the
  \emph{Dirichlet} and \emph{Neumann} restrictions
  \begin{equation}
  	\label{Dir-op}
 		\HHL_{D} := \sma{\smash{H^{(L)}_{D}}}{b} \qquad \text{and} \qquad
		\HHL_{N} := \sma{\smash{H^{(L)}_{N}}}{b}
	\end{equation}
	of $\HH$ to the $2|\Lambda_{L}|$-dimensional Hilbert space $\cH^{2}_{L} := \ell^{2}(\Lambda_{L}) \oplus \ell^{2}(\Lambda_{L})$. We also introduce
	the \emph{Dirichlet-bracketing} and \emph{Neumann-bracketing} restrictions
	\begin{equation}
 		\HHL_{+} := \tmatrix{\smash{H^{(L)}_{D}} & b}{b & -\smash{H^{(L)}_{N}}} \qquad \text{and} \qquad
		\HHL_{-} := \tmatrix{\smash{H^{(L)}_{N}} & b}{b & -\smash{H^{(L)}_{D}}}.
	\end{equation}
	In the above we have used the Dirichlet, resp.\ Neumann, restriction
\begin{equation}
 	H^{(L)}_{D/N} := -\Delta_{D/N}^{(L)} + U_{0} + V
\end{equation}
of $H$ to $\ell^{2}(\Lambda_{L})$, and all multiplication operators are to be interpreted as canonical restrictions.
\end{definition}

\begin{remarks}
 	\item The Neumann Laplacian on $\ell^{2}(\Lambda_{L})$ is also called \emph{graph Laplacian} or
		\emph{combinatorial Laplacian}. It is defined by
		\begin{equation}
 			(-\Delta_{N}^{(L)} \psi)(j) := \sum_{k\in \Lambda_{L}: |j-k|=1} [\psi(j) - \psi(k)]
		\end{equation}
		for every $\psi \in \ell^{2}(\Lambda_{L})$ and every $j\in\Lambda_{L}$.
		The Dirichlet Laplacian can be represented as the perturbation
		\begin{equation}
			\label{boundaries-relate}
 			-\Delta_{D}^{(L)} := -\Delta_{N}^{(L)} + 2\Gamma^{(L)},
		\end{equation}
		where $\Gamma^{(L)}$ is the multiplication operator given by
		\begin{equation}
 			(\Gamma^{(L)}\psi)(j) := \psi(j) \; \# \{k \in \ZZ^{d} \setminus \Lambda_{L} : |k-j| =1 \}
		\end{equation}
		for every $j\in\Lambda_{L}$. In other words, $\Gamma^{(L)}$ lives only on the outermost layer of $\Lambda_{L}$ and multiplies by the number of missing neighbours at every point.
	\item The finite-volume restrictions $-\Delta_{N}^{(L)}$ and $-\Delta_{D}^{(L)}$ are defined such as to obey Dirichlet-Neumann bracketing. Consequently we obtain the chains of inequalities
		\begin{equation}
 			\HHL_{-} \le \HHL_{N} \le \HHL_{+} \quad\text{and}\quad
				\HHL_{-} \le \HHL_{D} \le \HHL_{+}
		\end{equation}
		for every $\Lambda_{L} \subset \ZZ^{d}$ and the Dirichlet-Neumann bracketing properties
		\begin{equation}
 			\HH^{(\Lambda)}_{-} \oplus \HH^{(\Lambda')}_{-} \le \HH^{(\Lambda \cup \Lambda')}_{-}
			\quad\text{and}\quad
			\HH^{(\Lambda)}_{+} \oplus \HH^{(\Lambda')}_{+} \ge \HH^{(\Lambda \cup \Lambda')}_{+}
		\end{equation}
		for all disjoint $\Lambda, \Lambda' \subset \ZZ^{d}$ (in obvious abuse of our notation).
\end{remarks}

The desired interpretation of $\IZ$ follows from

\begin{lemma}
Given the random block operator $\HH$ from \eqref{ran-HH-diag} and any $E\in\RR$, we define the random, finite-volume eigenvalue counting function
\begin{equation}
 \IZ^{(L)}_{X}(E) := \frac{1}{2\,|\Lambda_{L}|} \; \tr_{\cH^{2}_{L}} \big[ \Chi_{]-\infty, E]}	(\HH_{X}^{(L)}) \big],
\end{equation}
where $X$ symbolises any self-adjoint restriction such that
\begin{equation}
	\label{interpol}
 	\HHL_{-} \le \HHL_{X} \le \HHL_{+}.
\end{equation}
Then
\begin{nummer}
\item there is a set $\Omega_{0} \subseteq \Omega$
	of full probability, $\PP(\Omega_{0})=1$, such that
	\begin{equation}
		\label{IDS-diag-limit}
 		\IZ(E) = \lim_{L\to\infty} \IZ^{(L)}_{X,\omega}(E)
	\end{equation}
	for every $\omega\in\Omega_{0}$, every continuity point $E\in\RR$ of $\IZ$ and every boundary condition $X$ satisfying \eqref{interpol}.
\item In addition we have the bounds
	\begin{equation}
		\label{Dir-Neu-ids}
 		\EE \big[ \IZ^{(L)}_{+}(E) \big]
  \le \IZ(E) \le \EE \big[  \IZ^{(L)}_{-}(E)\big]
	\end{equation}
	for every finite cube $\Lambda_{L}$ and every $E\in\RR$.
\end{nummer}
\end{lemma}

\begin{proof}
The arguments proceed as in the standard random Schr\"odinger case, see e.g.\ \cite{Kir89, CaLa90, PaFi92,KiMe07,Kir08}.
\end{proof}


\section{Random diagonal $B$: boundedness of the density of states}
\label{sec:wegner}


We have seen in Section~\ref{constndiag} that the \emph{density} of states of $\HH$ with off-diagonal operators $B$ that are a constant multiple of the identity can
have a singularity, even if the density of states of $H$ is bounded.
In this section we consider random operators $\HH$ of the form \eqref{ran-HH-diag} with diagonal disorder in all blocks.
We will prove a Wegner estimate for such models. It implies that the density of states of $\HH$ exists and is bounded provided $H$ or $b$ is bounded away from zero and the distribution of the random variables is absolutely continuous with a suitably regular Lebesgue density. Technically, one has to cope with a non-monotone dependence on the random potential: the single-site potential that enters $\HH$ is sign-indefinite and has mean zero. In contrast, for ordinary (i.e.\ non-block) random Schr\"odinger operators with such single-site potentials it is not known how to prove a comparable Wegner estimate. Existing Wegner estimates involve either higher powers of the volume \cite{Ves10} or yield only log-H\"older continuity of the integrated density of states \cite{Kru10}.

\begin{figure}
  \begin{center}
    \begin{tabular}{@{}cc@{}}
      \resizebox{61mm}{!}{\includegraphics{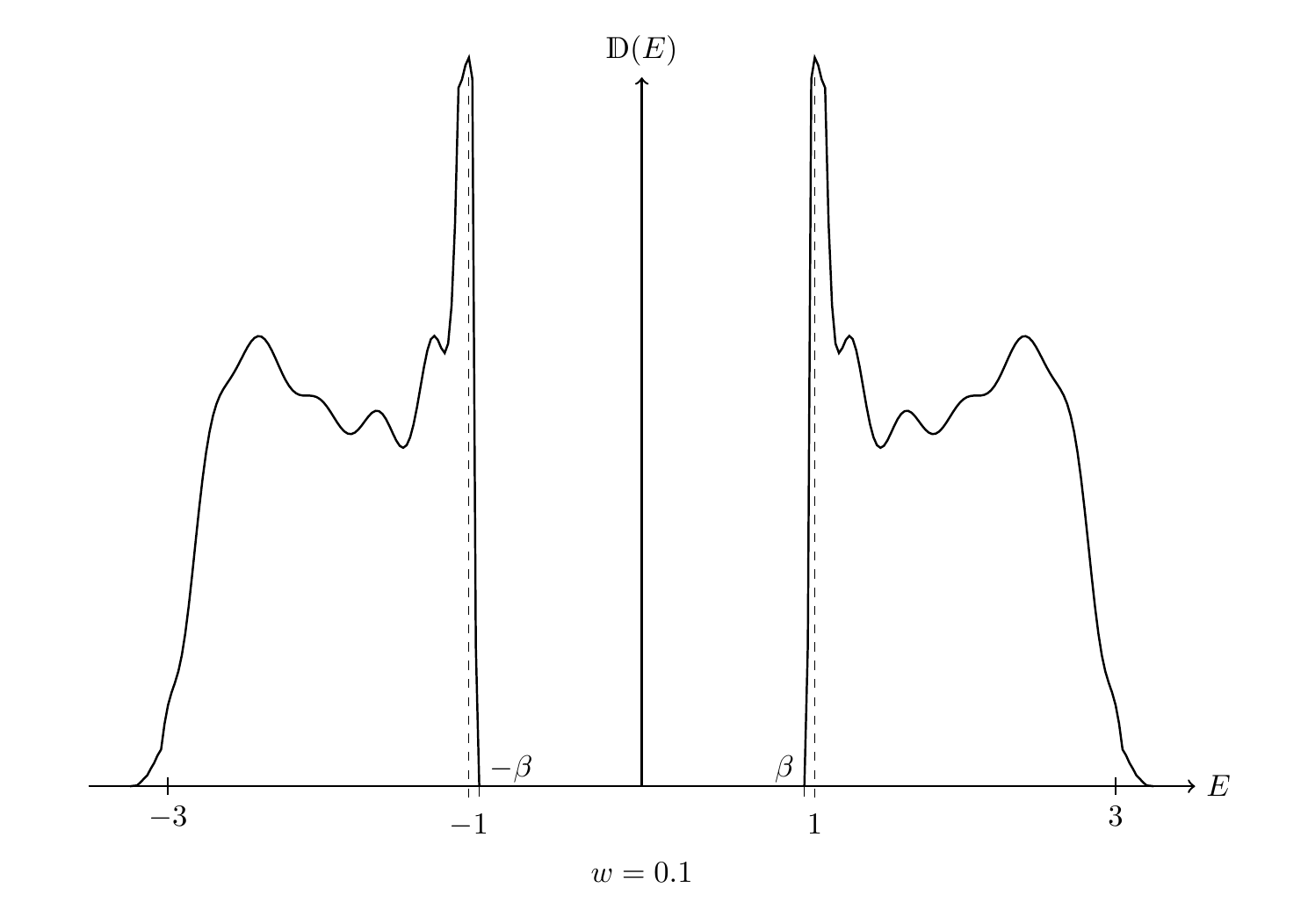}} &
      \resizebox{61mm}{!}{\includegraphics{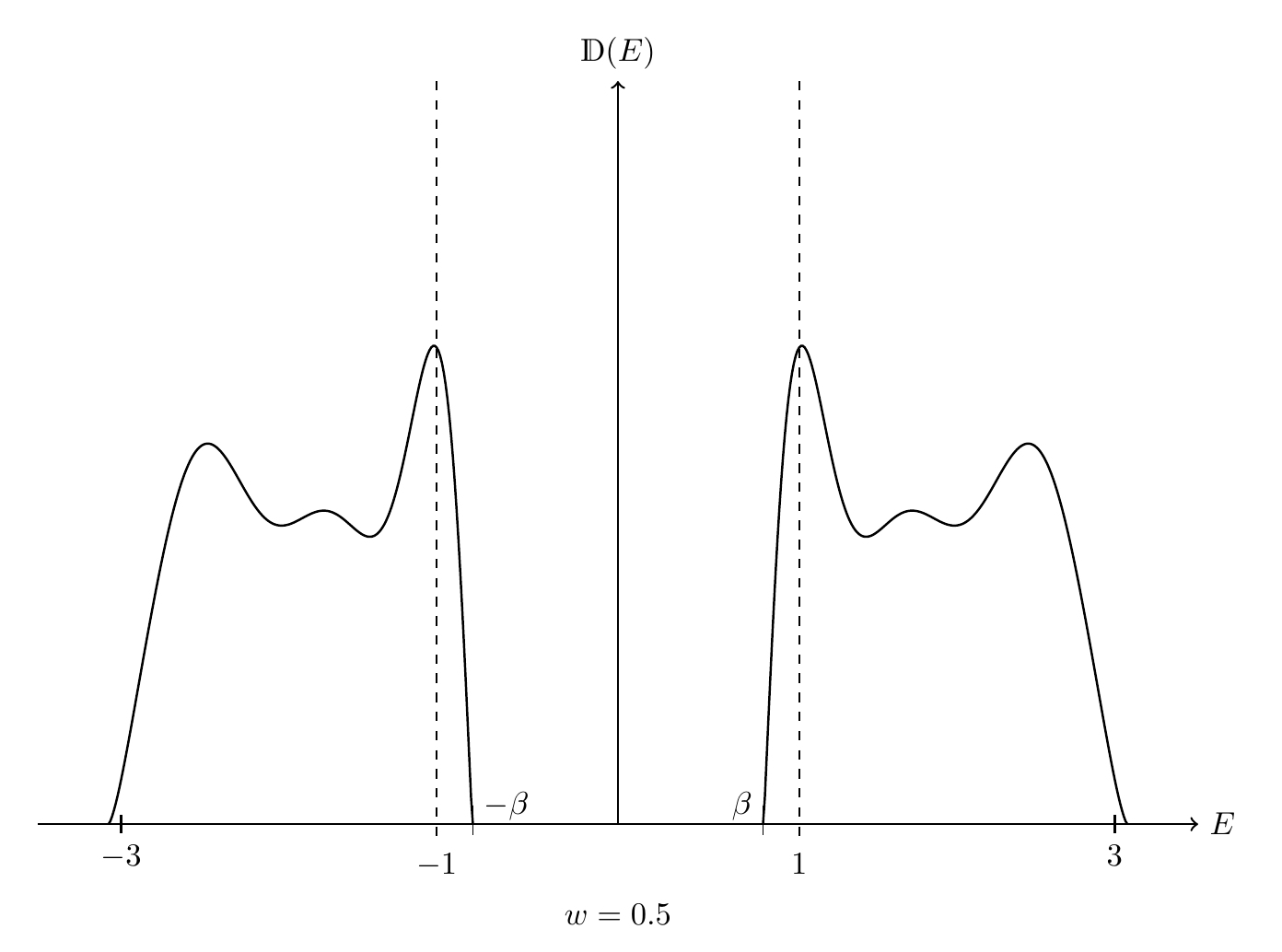}} \\
      \resizebox{61mm}{!}{\includegraphics{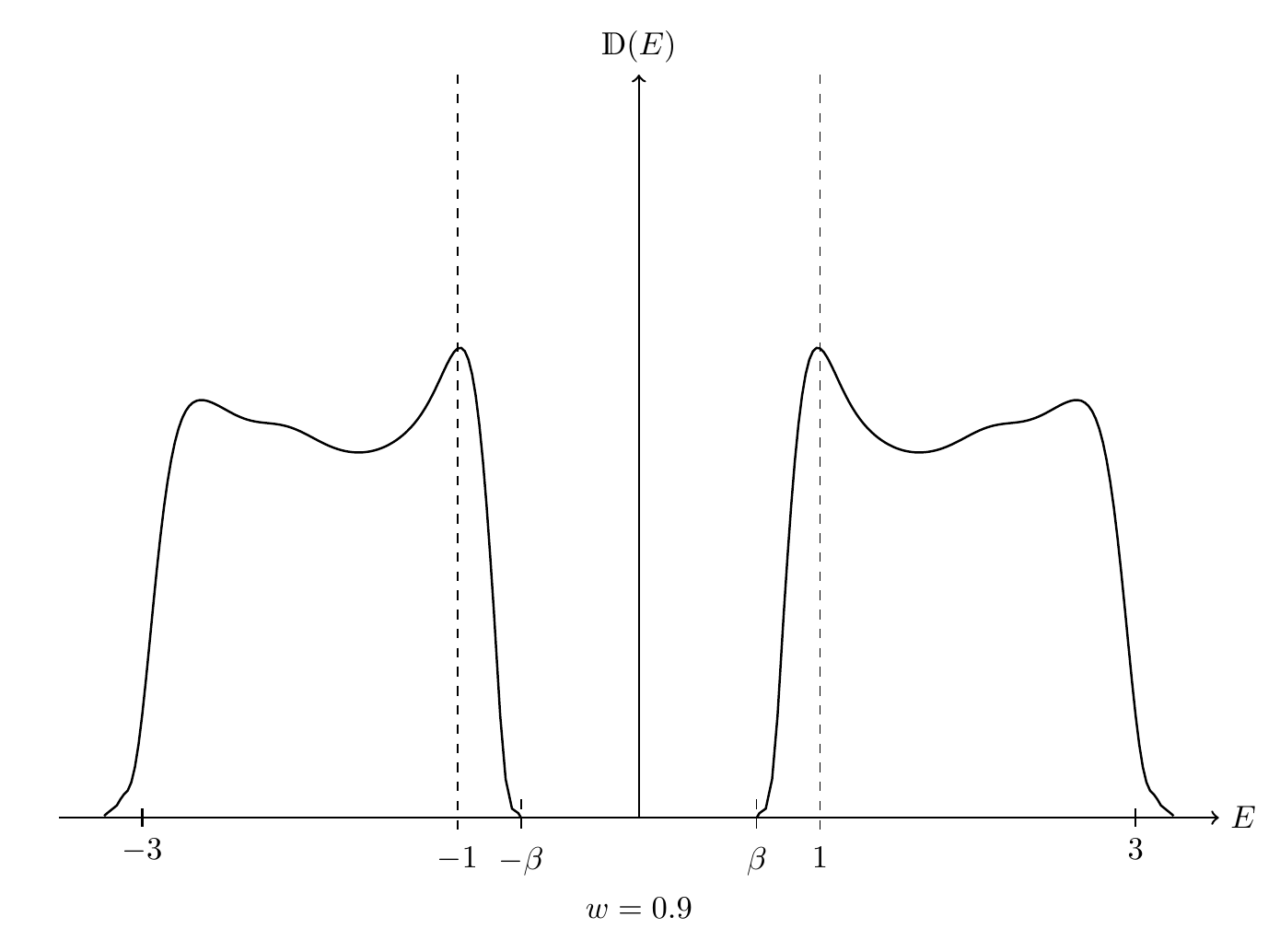}} &
      \resizebox{61mm}{!}{\includegraphics{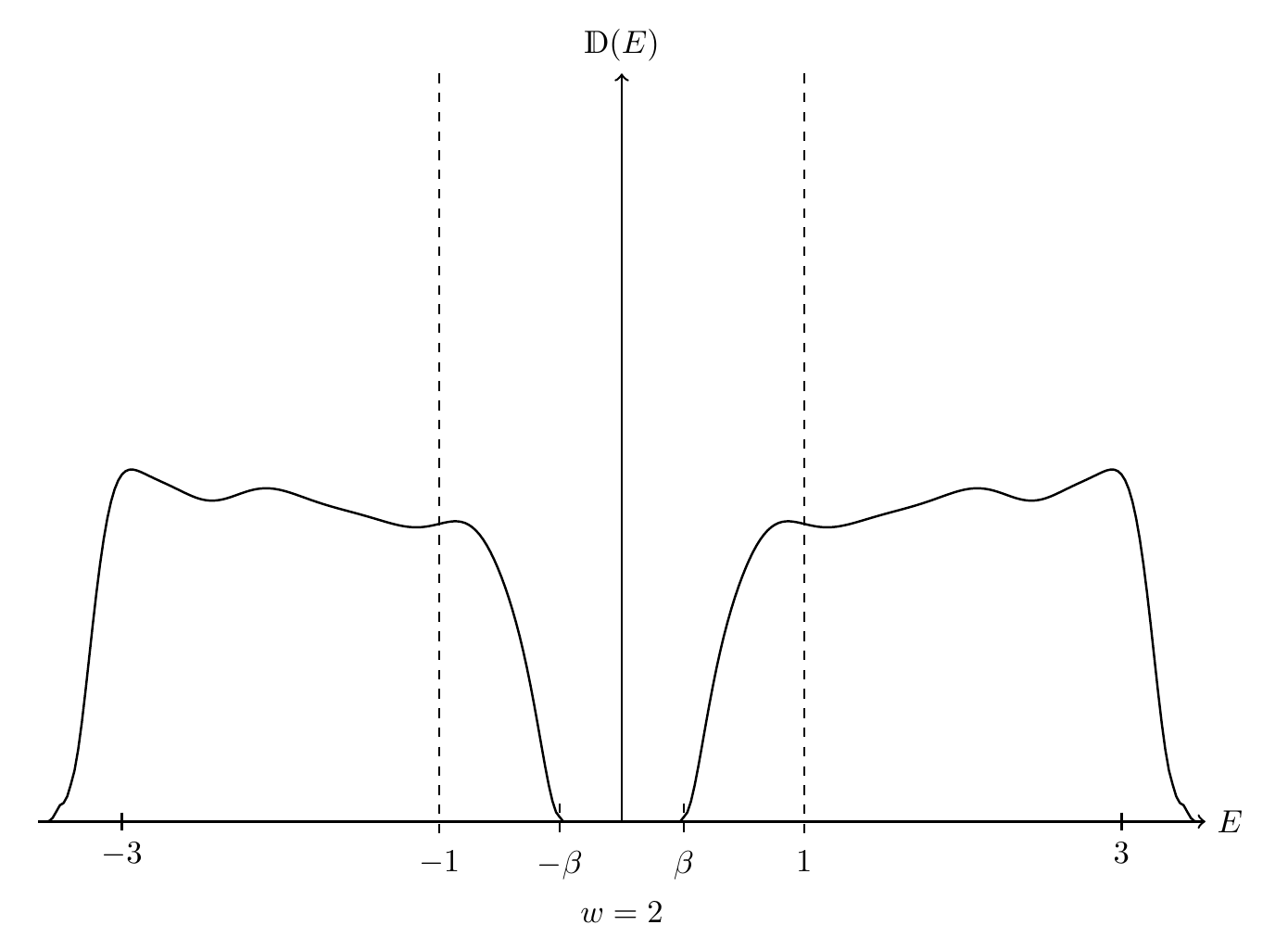}}
    \end{tabular}
    \caption{Assuming \protect{$b_{\omega}(j)=1+ w \:\wtilde{b}_{\omega}(j)$} with $\wtilde{b}_{\omega}(j)$ distributed uniformly in the interval $[-0.5,0.5]$ we display the dependence of the density of states on the disorder parameter $w$. (Assumption~(B) of Theorem~\ref{Wegner} is satisfied for $w=0.1,0.5,0.9$).}
    \label{fig:smooth}
  \end{center}
\end{figure}

\begin{theorem}\label{Wegner}
Consider the random block operator
\begin{equation}
\HH : \omega\mapsto \HH_{\omega} =\sma{H_{\omega}}{b_{\omega}}
\end{equation}
where $H$ and $b$ are given as in \eqref{ran-Schroed} and \eqref{ran-offdiag}. Assume that at least one of the following two conditions is satisfied:
\par\smallskip\noindent
\textup{(H)}  \quad  \parbox[t]{11.5cm}{there exists $\lambda>0$ such that $H\ge \lambda 1$ holds $\PP$-a.s.\ and $\mu_V$ is absolutely continuous with a piecewise continuous Lebesgue density $\phi_{V}$ of bounded variation and compact support,}   \\[1ex]
\textup{(B)} \quad  \parbox[t]{11.5cm}{there exists $\beta>0$ such that $b\ge \beta 1 $ holds
$\PP$-a.s.\ and $\mu_b$ is absolutely continuous with a piecewise continuous Lebesgue density $\phi_{b}$ of bounded variation and compact support.}
\par\smallskip\noindent
Then the integrated density of states $\IZ$ of $\HH$ is Lipschitz continuous and has a bounded density $\DD:= \d\IZ/\d E$. In particular, we have for Lebesgue-almost all $E\in\RR$
\begin{equation}
\DD(E) \leq 2\,\frac{|E| + 1}{\lambda}\, \|\phi_V\|_{\textsc{bv}},
\end{equation}
if Condition~\textup{(H)} applies, and
\begin{equation}
\DD(E) \leq 2\,\frac{|E| + 1}{\beta}\,  \|\phi_b\|_{\textsc{bv}} ,
\end{equation}
if Condition~\textup{(B)} applies. Here, $\|f\|_{\textsc{bv}}$ denotes the total variation of $f:\RR \rightarrow \CC$.
\end{theorem}

\begin{remarks}
\item \label{symm-dos}
	The density of states $\DD$ is an even function on account of Lemma~\ref{symm-spec}.
\item
We recall that, typically, $\DD$ exhibits a singularity at the inner band edges for constant off-diagonal blocks $b= \beta 1 >0$, see the left panel of Figure~\ref{Figure1}. Case~(B) of Theorem~\ref{Wegner} implies that, regardless of $H$, this singularity is smeared out by the disorder. In particular, $\DD$ remains bounded.
\item Figure~\ref{fig:smooth} displays a typical density of states for random block operators in the case~(H) of Theorem~\ref{Wegner}. Here, $\DD$ is not only bounded but even vanishes at the inner band edges, because it exhibits a Lifshits tail, see Section~\ref{sec:lif}.
\end{remarks}


\begin{proof}[Proof of Theorem \ref{Wegner}]
As compared to the proof of the Wegner estimate for the standard Anderson model, there are essentially two modifications necessary here. They have been isolated in the subsequent Lemmas~\ref{wegnerhilf} and~\ref{wegner-hilf-2}.

Lemma~\ref{diag-offdiag} tells us that the roles of the diagonal operator $H$ and of the
off-diagonal operator $b$ are interchangeable. Thus it suffices to prove only the
case (H). Moreover, because of Lemma~\ref{symm-spec} we will restrict ourselves to energies $E \ge 0$ without loss of generality. In fact, using Condition~(H) and Proposition~\ref{gap-H} it then suffices to consider $E \ge \lambda$, which we will do from now on. We denote by $\HHL(\Vv) :=\HHL_{D}$ the Dirichlet restriction \eqref{Dir-op}
 of the operator $\HH$ to the cube $\Lambda_L$. In this notation we make
explicit the dependence of the operator on the random variables
$\Vv :=\{V_{j}\}_{j\in\Lambda_{L}}$ (we prefer to write $V_{j}$ instead of $V(j)$ in this proof and the subsequent lemmas) and suppress the dependence
on the entries of the off-diagonal blocks.
Furthermore, we write $E_n(\Vv) := E_n\big(\HHL(\Vv)\big)$ for the
$n^{\mathrm{th}}$ eigenvalue of $\HHL(\Vv)$, where the eigenvalues
are ordered by magnitude and repeated according to multiplicity.
Finally, we fix $\varepsilon \in ]0,\min\{\lambda,1\}/3[$ and consider a switch function $\varrho$,
i.e.\ $\varrho \in C^{1}(\RR)$ is  non-decreasing with $0\leq\varrho \leq 1$,
$\varrho(\eta)=1$ for $\eta > \varepsilon$ and $\varrho(\eta)=0$ for $\eta <-\varepsilon$. Then
\begin{equation}
0 \leq \Chi_{]E-\varepsilon, E+\varepsilon]}(\eta) \leq
\varrho(\eta-E+2\varepsilon)-\varrho(\eta-E-2\varepsilon)
\end{equation}
for all $\eta\in\RR$, whence
\begin{align}
\label{switch-rewrite}
\tr_{\cH^{2}_{L}} \Big[ &\Chi_{]E -\varepsilon, E+\varepsilon]} \big(\HHL(\Vv)\big) \Big]  \nonumber\\
&\le  \sum_{n=1}^{2|\Lambda_{L}|} \Big[ \varrho\big(E_n(\Vv) - E + 2\varepsilon \big) -
		\varrho\big(E_n(\Vv) - E - 2\varepsilon \big) \Big] \nonumber\\
&= - \sum_{n=1}^{2|\Lambda_{L}|} \int_{E-2\varepsilon}^{E+2\varepsilon} \frac{\partial}{\partial \eta}
				\varrho\big(E_n(\Vv) - \eta\big)\,\d\eta \nonumber \\
&= \sum_{n=1}^{2|\Lambda_{L}|} \int_{E-2\varepsilon}^{E+2\varepsilon}
			\varrho'\big(E_n(\Vv) - \eta\big)\,\d\eta .
\end{align}
The chain rule tells us that
\begin{equation}
\sum_{j\in\Lambda_L} \frac{\partial}{\partial V_j}\,\varrho\big(E_n(\Vv) - \eta\big) =
\varrho'\big(E_n(\Vv) - \eta\big) \sum_{j\in\Lambda_L} \frac{\partial E_n(\Vv)}{\partial V_j}.
\end{equation}
Since $H^{(L)}(\Vv) \ge \lambda >0$ by Condition~(H), we conclude from this identity and
Lemma~\ref{wegnerhilf}
\begin{align}
\label{ev-move}
\varrho'\big(E_n(\Vv) - \eta\big)
&\le \frac{ E_n(\Vv)}{\lambda} \sum_{j\in\Lambda_L} \frac{\partial}{\partial V_j}\,\varrho\big(E_n(\Vv) - 	\eta\big) \nonumber\\
&\le \frac{ E + 1}{\lambda} \sum_{j\in\Lambda_L} \frac{\partial}{\partial V_j}\,\varrho\big(E_n(\Vv) - 	\eta\big)
\end{align}
for all $n\in\NN$ and all $\eta \in [E-2\varepsilon, E+2\varepsilon]$. We note that the last inequality uses  $3 \varepsilon < \min\{\lambda,1\}$. This guarantees that only those $n$ with $E_{n}(\Vv) \in ]0, E+1[$ contribute and that the $j$-sum is non-negative for these $n$ by Lemma~\ref{wegnerhilf}.

Since the random variables $\{V_j\}_{j\in\Lambda_{L}}$ are independent and have
the same individual distribution $\mu_{V}$, the expectation  is just
integration with respect to the product of these distributions. Thus \eqref{switch-rewrite} and \eqref{ev-move} imply
\begin{align}
	\label{wegner2}
	& \EE\Big\{  \tr_{\cH^{2}_{L}} \big[ \Chi_{]E -\varepsilon, E+\varepsilon]} (\HHL_{D}) \big]    \Big\} \notag\\
	&\le \frac{E+1}{\lambda} \sum_{j\in\Lambda_L} \int_{E-2\varepsilon}^{E+2\varepsilon}
		\int_\RR\!\!\ldots\int_\RR \left(  \int_\RR \frac{\partial}{\partial V_j}
		\sum_{n=1}^{2|\Lambda_{L}|} \varrho\big(E_n(\Vv)-\eta\big)\, \d\mu_{V}(V_j)\right) \notag\\
	& \hspace*{5.5cm}\times\;\bigg(\prod_{k\in\Lambda_{L}: \; k\not=j}\d\mu_{V}(V_{k})\bigg)\,\d\eta.
\end{align}
Since, in general, the function
\begin{equation}
	V_{j} \mapsto F(V_{j}) :=  \sum_{n=1}^{2|\Lambda_{L}|} \varrho\big(E_n(\Vv)-\eta\big)
	= \tr_{\cH^{2}_{L}} \Big[ \varrho\big(\HHL(\Vv) -\eta\big) \Big]
\end{equation}
is non-monotone in its argument for given $\eta\in\RR$ and $V_{k}\in\RR$, $k\neq j$, we deviate from the standard reasoning at this point. Clearly $F\in C^{1}(\RR)$ by analytic perturbation theory. Moreover changing $V_{j}$ amounts to a rank-2-perturbation of $\HHL(\Vv)$. Thus we have $|F(v) - F(v')| \le 2$ for all $v,v' \in\RR$ and we can apply Lemma \ref{wegner-hilf-2} to \eqref{wegner2}. This gives
\begin{equation}
 \EE\Big\{  \tr_{\cH^{2}_{L}} \big[ \Chi_{]E -\varepsilon, E+\varepsilon]} (\HHL_{D}) \big]    \Big\} \le 8 \varepsilon \, L^{d} \|\phi_{V}\|_{\textsc{bv}}
 \, \frac{E+1}{\lambda}.
\end{equation}
Finally, the assertion follows from Lemma~\ref{IDS-diag-limit} and dominated convergence.
\end{proof}

We come to the main deterministic tool used in the proof of Theorem \ref{Wegner}. Given Assumption~(H), it ensures that eigenvalues move around strongly enough when the random variables change their values.

\begin{lemma}
	\label{wegnerhilf}
	Let $E(\Vv)$ be an eigenvalue of $\sma{H^{(L)}(\Vv)}{B^{(L)}}$. Then
\begin{equation}
E(\Vv) \sum_{j\in\Lambda_L} \frac{\partial E(\Vv)}{\partial V_j} \geq \inf\spec\big(H^{(L)}(\Vv)\big).
\end{equation}
\end{lemma}

\begin{proof}
Let $\Psi(\Vv) \equiv \Psi=\binom{\psi_{1}}{\psi_{2}} \in \cH^{2}_{L}$ be an eigenvector corresponding to the eigenvalue $E(\Vv) \equiv E$, normalised according to
$\llangle\Psi,\Psi\rrangle =\langle\psi_1,\psi_1\rangle
+\langle\psi_2,\psi_2\rangle = 1$. Thus, writing $H^{(L)}(\Vv) \equiv H^{(L)}$, it satisfies
\begin{equation}
	\label{h1-2}
	\begin{split}
		H^{(L)}\psi_1 + B^{(L)}\psi_2 &= E \psi_1, \\
		B^{(L)}\psi_1 - H^{(L)}\psi_2 &= E \psi_2.
	\end{split}
\end{equation}
By the Feynman-Hellmann theorem (see e.g.\ \cite{RS4}) we have
\begin{equation}
	\frac{\partial E}{\partial V_j} = |\psi_1(j)|^2-|\psi_2(j)|^2
\end{equation}
for every $j\in\Lambda_{L}$ and consequently, inserting \eqref{h1-2} we find
\begin{align}
E \sum_{j\in\Lambda_L} \frac{\partial E}{\partial V_j}
&= E \langle \psi_1,\psi_1\rangle - E \langle \psi_2,\psi_2\rangle \nonumber\\[-1.5ex]
&=\langle\psi_1,H^{(L)}\psi_1 + B^{(L)}\psi_2\rangle -
	\langle B^{(L)}\psi_1 - H^{(L)}\psi_2,\psi_2\rangle \nonumber\\
&= \langle\psi_1,H^{(L)}\psi_1\rangle + \langle \psi_2,H^{(L)}\psi_2\rangle \nonumber\\
&\ge \inf\spec\big(H^{(L)}\big).
\end{align}
\end{proof}

The next lemma deals with the problem of the non-monotonous dependence of the cumulative eigenvalue counting function on the random potential. It is here where we have to assume a suitable regularity of the  distribution.

\begin{lemma} \label{wegner-hilf-2}
Let $\phi: \RR \rightarrow\CC$ be piecewise continuous, of bounded variation and have compact support. Let $F \in C^{1}(\RR)$ and assume the existence of a constant $a>0$ such that $|F(x) - F(y)| \le a$ for all $x,y\in\RR$.
Then we have
\begin{equation}
\label{bv}
\left|\int_\RR F'(x) \, \phi(x)\,\d x\right| \le a \, \|\phi \|_{\textsc{bv}}.
\end{equation}
\end{lemma}

\begin{proof}
Step 1:~~ We prove the claim for step functions
\begin{equation}
 \phi_{N} = \sum_{\nu =1}^{N} \xi_{\nu} \Chi_{]x_{\nu-1}, x_{\nu}]} = \sum_{\nu =1}^{N} (\xi_{\nu} - \xi_{\nu -1}) \Chi_{]x_{\nu-1}, x_{N}]}
\end{equation}
where $N \in\NN$, $\xi_{\nu} \in \CC$ for all $\nu \in \{1, \ldots, N\}$, $\xi_{0} :=0$ and $-\infty < x_{0} < x_{1} <  \ldots < x_{N} <  \infty$. For such $\phi_{N}$ we get
\begin{align}\label{wegner1}
\left| \int_\RR \! F'(x) \, \phi_{N}(x)\,\d x \right|
&= \bigg| \sum_{\nu=1}^{N} (\xi_{\nu} - \xi_{\nu-1} ) [F(x_{N}) - F(x_{\nu -1})]\bigg| \nonumber\\
&\le a \sum_{\nu=1}^{N} |\xi_{\nu} - \xi_{\nu-1}| \le a\, \|\phi_{N}\|_{\textsc{bv}}.
\end{align}

Step 2:~~ For general $\phi$ as in the lemma, the claim \eqref{bv} follows from a uniform approximation of $\phi$ by step functions of the form
$\phi_{N} := \sum_{\nu=1}^{N} \phi(x_{\nu}) \Chi_{]x_{\nu-1}, x_{\nu}]}$. Indeed,
for every given $\varepsilon>0$ we can choose discretisation points $x_{\nu}$, $\nu=1,\ldots,N$ and $N\in\NN$ such that $\|\phi - \phi_{N}\|_{\infty} \le \varepsilon$, because $\phi$ is piecewise uniformly continuous. Therefore $\int_{\RR} F'(x) \, \phi(x)\,\d x = \lim_{N\to\infty} \int_{\RR} F'(x) \, \phi_{N}(x)\,\d x $ for a suitable sequence of step functions $\phi_{N}$, because $F'$ is bounded on the support of $\phi$. This,
the bound from Step~1 and $\|\phi_{N}\|_{\textsc{bv}} \le \|\phi\|_{\textsc{bv}}$ finish the proof.
\end{proof}

\section{Random diagonal $B$: Lifshits tails}
\label{sec:lif}

It is a striking fact that the integrated density of states of random Schr\"odinger operators grows only exponentially slowly in the vicinity of fluctuation band edges. This behaviour is called \textit{Lifshits tail} \cite{Lif64}.

In Theorem~\ref{lif-gen} we provide a result for random block operators $\HH$ of the form \eqref{ran-HH-diag} which limits the growth of the integrated density of states $\IZ$ at energy $\lambda$, provided $H \ge \lambda 1 >0$ exhibits a Lifshits tail at energy $\lambda$. We emphasise that Theorem~\ref{lif-gen} is only interesting in the case where $\pm\lambda$ coincide with the endpoints of the spectral gap of $\HH$ around zero. According to Corollary~\ref{gap-end} this always happens if $0 \in\supp(\mu_{b})$.

\begin{theorem}
	\label{lif-gen}
Consider the random block operator
\begin{equation}
\HH : \omega\mapsto \HH_{\omega} =\sma{H_{\omega}}{b_{\omega}}
\end{equation}
where $H$ and $b$ are given as in \eqref{ran-Schroed} and \eqref{ran-offdiag} and let $\IZ$ be its integrated density of states. Suppose in addition that $\lambda := \inf\spec(H) > 0$ and that there exists constants $\alpha,\gamma >0$ such that for all sufficiently small $\varepsilon >0$
\begin{equation}
 \label{lif-H}
 \PP\big[ \inf\spec(H^{(L_{\varepsilon})}_{N}) \le \lambda + \varepsilon \big]  \le \e^{-\gamma\varepsilon^{-\alpha}}
\end{equation}
where $L_{\varepsilon}$ is a sequence of diverging lengths as $\varepsilon\downarrow 0$ with $\lim_{\varepsilon\downarrow 0} \varepsilon^{\alpha/d} L_{\varepsilon}$ exists and lies in $]0,\infty[$.
Then the estimate
	\begin{equation}
		\label{upper-ineq}
 		\limsup_{\varepsilon \downarrow 0} \frac{\ln \big| \ln [ \IZ(\lambda + \varepsilon) - \IZ(\lambda)] \big|}{\ln \varepsilon} \le -\alpha
	\end{equation}
holds.
\end{theorem}

\begin{remarks}
	\item Assumption \eqref{lif-H} is the statement which is typically proven for a random Schr\"odinger operators $H$ of the form \eqref{ran-Schroed} when establishing the upper bound
	\begin{equation}
 		\limsup_{\varepsilon \downarrow 0} \frac{\ln \big| \ln [ N(\lambda + \varepsilon) ] \big|}{\ln \varepsilon} \le -\alpha
	\end{equation}
	for a Lifshits tail with Lifshits exponent $\alpha =d/2$ of the integrated density of states $N$ at the lower spectral edge $\lambda$. In other words, if $0\in\supp(\mu_{b})$, then Theorem~\ref{lif-gen} says that the growth of $\IZ$ near the lower edge of the positive a.s.\ spectrum of $\HH$ is no faster than the growth of $N$ near the bottom of the a.s.\ spectrum of $H$. Note that the bound on $\IZ$ holds independently of $b$.
  \item An analogous statement to Theorem~\ref{lif-gen} holds at the upper edge of the negative spectrum of $\HH$.
\end{remarks}

\begin{proof}[Proof of Theorem~\ref{lif-gen}]
We claim that for all energies $E \ge 0$ and all finite cubes $\Lambda_{L}$ the estimate
\begin{equation}
 	\IZ(E) - \IZ(0) \le \EE\big[ \IZ_{-}^{(L)}(E) - \IZ_{-}^{(L)}(0)\big]
\end{equation}
holds. This follows from  Dirichlet-Neumann bracketing \eqref{Dir-Neu-ids} and
\begin{equation}
	\label{zero}
	\IZ(0) = \frac{1}{2} = \EE \big[ \IZ_{N}^{(L)}(0) \big] =\EE \big[ \IZ_{-}^{(L)}(0) \big].
\end{equation}
The first two equalities in \eqref{zero} are based on Lemma~\ref{symm-spec} and that zero is not in the spectrum. In order to see the last equality in \eqref{zero} we view $\HHL_{-} = \HHL(1)$ as an analytic perturbation of $\HHL_{N} = \HHL(0)$, where
\begin{equation}
  \HHL(a) := \HHL_{N} -2 a \tmatrix{0 & 0}{0 & {\Gamma^{(L)}}}
\end{equation}
for $a\in\RR$ and $\Gamma^{(L)}$ was introduced in \eqref{boundaries-relate}. Analytic perturbation theory tells us that the eigenvalues of $\HHL(a)$ depend continuously on the parameter $a$. On the other hand, we infer from Lemma~\ref{gap-different} that zero lies in an open spectral gap of $\HHL(a)$ of size at least $2\lambda$ for every $a\in [0,1]$. Therefore $\HHL(0)$ must have exactly as many positive (negative) eigenvalues as $\HHL(1)$, and the third equality in \eqref{zero} holds.

Now let $\varepsilon >0$ and observe
\begin{equation}
	\label{lif-start}
  \EE\big[ \IZ_{-}^{(L)}(\lambda +\varepsilon) - \IZ_{-}^{(L)}(0)\big] \le  \; \PP \big[ \spec(\HHL_{-}) \cap ]0, \lambda +\varepsilon] \neq \emptyset \big].
\end{equation}
By contradiction we conclude from Lemma~\ref{gap-H} that the event in the probability on the r.h.s.\ of \eqref{lif-start} implies the event $\inf \spec(H_{N}^{(L)}) \le \lambda + \varepsilon$. Hence, we get
\begin{equation}
 	\IZ(\lambda +\varepsilon) - \IZ(0) \le \PP \big[ \inf \spec(H_{N}^{(L)}) \le \lambda + \varepsilon \big],
\end{equation}
and the claim follows from \eqref{lif-H}.
\end{proof}


\begin{thebibliography}{MMM}
\frenchspacing

\bibitem[AlSZ]{ASZ02}
  \au{A.}{Altland},\au{B. D.}{Simons}\et\lau{M.}{Zirnbauer}
  \ti{Theories of low-energy quasi-particle states in disordered d-Wave superconductors}
  \z{Phys. Rep.}{359}{283--354}{2002}

\bibitem[AlZ]{AZ97}
  \au{A.}{Altland}\et\lau{M.}{Zirnbauer}
  \ti{Nonstandard symmetry classes in mesoscopic normal-superconducting hybrid structures}
  \z{Phys. Rev. B}{55}{1142--1161}{1997}

\bibitem[An]{And59}
  \lau{P. W.}{Anderson}
  \ti{Theory of dirty superconductors}
  \z{J. Phys. Chem. Solids}{11}{26--30}{1959}

\bibitem[BVZ]{BaVeZh06}
  \au{A. V.}{Balatsky}, \au{I.}{Vekhter}\et\lau{J. X.}{Zhu}
  \ti{Impurity-induced states in conventional and unconventional superconductors}
  \z{Rev. Mod. Phys.}{78}{373--433}{2006}

\bibitem[CL]{CaLa90}
  \au{R.}{Carmona}\et\lau{J.}{Lacroix}
  \bti{Spectral theory of random Schr\"odinger operators}
  \pub{Birk\-h\"auser}{Boston}{1990}

\bibitem[dG]{deGe66}
  \lau{P. G.}{de Gennes}
  \bti{Superconductivity of Metals and Alloys}
  \pub{Benjamin}{New York}{1966}

\bibitem[DL]{DuLe00}
  \au{A. C.}{Durst}\et\lau{P. A.}{Lee}
  \ti{Impurity-induced quasiparticle transport and universal-limit Wiedemann-Franz violation
			in d-wave superconductors}
  \z{Phys. Rev. B}{62}{1270-1290}{2000}

\bibitem[HM]{HisMu08}
  \au{P.}{Hislop} \et\lau{P.}{M\"uller}
  \ti{A lower bound for the density of states of the lattice Anderson model}
  \z{Proc. Amer. Math. Soc}{136}{2887-2893}{2008}

\bibitem[J]{Je92}
  \lau{F.}{Jeske}
  \bti{\"Uber lokale Positivit\"at der Zustandsdichte zuf\"alliger Schr\"odinger-Ope\-ra\-to\-ren}
  \pub{Ph.D. thesis}{Ruhr-Universit\"at Bochum}{1992}

\bibitem[Ki1]{Kir89}
  \lau{W.}{Kirsch}
  \ti{Random Schr\"odinger operators: a course}
  In
  \au{H.}{Holden}\et\au{A.}{Jensen} (Eds.),
  \bti[Lecture Notes in Physics 345]{Schr\"odinger operators}
  \pub[pp. 264--370]{Springer}{Berlin}{1989}

\bibitem[Ki2]{Kir08}
  \lau{W.}{Kirsch}
  \ti{An invitation to random Schr\"odinger operators}
	\z{Panoramas et Synth\`eses}{25}{1--119}{2008}
	
\bibitem[KM]{KiMe07}	
	\au{W.}{Kirsch} \et\lau{B.}{Metzger}
	\ti{The integrated density of states for random Schr\"o-dinger operators}
	 In
	 \au{F.}{Gesztesy}, \au{P.}{Deift}, \au{C.}{Galvez}, \au{P.}{Perry} \et\au{W.}{Schlag} (Eds.),
   \bti[Proc. Sympos. Pure Math. 76, Part 2]{Spectral theory and mathematical physics: A Festschrift in honor of Barry Simon's 60th birthday}	
	 \pub[pp. 649--696]{Amer. Math. Soc.}{Providence, RI}{2007}

\bibitem[Kr]{Kru10}
  \lau{H.}{Kr\"uger}
  \ti{Localization for random operators with non-monotone potentials with exponentially
			decaying correlations}
	preprint arXiv:1006.5233, 2010.


\bibitem[L]{Lif64}
	\lau{I. M.}{Lifshitz}
	\ti{The energy spectrum of disordered systems}
	\z{Adv. Phys.}{13}{483--536}{1964}


\bibitem[PF]{PaFi92}
  \au{L.}{Pastur}\et\lau{A.}{Figotin}
  \bti{Spectra of random and almost-periodic operators}
  \pub{Springer}{Berlin}{1992}

\bibitem[RS]{RS4}
    \au{M.}{Reed}\et\lau{B.}{Simon}
    \bti{Methods of modern mathematical physics IV: analysis of operators}
    \pub{Academic Press}{San Diego}{1978}

\bibitem[T]{Tre08}
  \lau{C.}{Tretter}
  \bti{Spectral theory of block operator matrices and applications}
  \pub{Imperial College Press}{London}{2008}

\bibitem[Ve]{Ves10}
	\lau{I.}{Veseli\'c}
	\ti{Wegner estimate for discrete alloy-type models}
	Preprint 	arXiv:\linebreak[1]1006.4995, 2010, to appear in Ann. H. Poincar\'e.

\bibitem[ViSF]{ViSeFi00}
  \au{S.}{Vishveshwara}, \au{T.}{Senthil}\et\lau{M. P. A.}{Fisher}
  \ti{Superconducting ``metals'' and ``insulators''}
  \z{Phys. Rev. B}{61}{6966--6981}{2000}

\bibitem[W]{We91}
	\lau{F.}{Wegner}
	\ti{Bounds on the density of states in disordered systems}
	\z{Z. Phys. B}{44}{9--15}{1981}

\bibitem[Z]{Zie92}
	\lau{K.}{Ziegler}
	\ti{Quasiparticle states in disordered superfluids}
	\z{Z. Phys. B}{86}{33--38}{1992}

\end{thebibliography}
\end{document}